\documentclass[10pt]{article}
\usepackage[utf8]{inputenc}
\usepackage[fleqn]{amsmath}
\usepackage{amsthm}
\usepackage{amsfonts}
\usepackage{amssymb}
\usepackage[left=80pt,right=80pt,top=60pt,bottom=70pt]{geometry}
\usepackage{graphicx}
\usepackage{tikz}
\usetikzlibrary{decorations.markings}
\usetikzlibrary{arrows.meta}
\usepackage{tikz}
\usetikzlibrary{shapes.misc}
\usepackage{color}
\usepackage[bookmarks]{hyperref}

\newtheorem{definition}{Definition}
\newtheorem{lemma}{Lemma}
\newtheorem{theorem}{Theorem}
\newtheorem{corollary}{Corollary}

\newlength{\Mdent}
\newcommand*{\Mdentspace}{\hspace{\Mdent}}
\newcommand*{\noMdent}{\setlength{\mathindent}{0pt}}
\newcommand*{\reMdent}{\setlength{\mathindent}{\Mdent}}

\newcommand*{\IMP}{\makebox[\Mdent][l]{$\Longrightarrow$}}
\newcommand*{\IFF}{\makebox[\Mdent][l]{$\Longleftrightarrow$}}

\newenvironment*{listeqn}{
\newcommand*{\COM}{\,,\\ & &&}
\newcommand*{\AND}{\\ &\text{and}&&}
\noMdent \alignat{2} &\Mdentspace && }{\endalignat}

\newcommand*{\BY}[1]{\qquad \text{by \eqref{#1}}}
\newcommand*{\BYthm}[1]{\qquad \text{by Theorem \ref{#1}}}
\newcommand*{\BYlem}[1]{\qquad \text{by Lemma \ref{#1}}}

\newcommand*{\st}{\text{ s.t. }}
\newcommand*{\Z}[1]{\mathbb{Z}_{#1}}
\newcommand*{\setdef}[3][]{\mathopen #1 \lbrace #2  \mathrel #1 \vert  #3 \mathclose #1 \rbrace}
\newcommand*{\set}[2][]{\mathopen #1 \lbrace #2 \mathclose #1 \rbrace}

\newcommand*{\MinSupp}[2][]{\mathrm{MinSupp}\mathopen #1 ( #2 \mathclose #1 )}
\newcommand*{\supp}[2][]{\mathrm{supp}\mathopen #1 ( #2 \mathclose #1 )}
\newcommand*{\dif}[2]{ \frac{\partial #1}{\partial #2} }

\DeclareRobustCommand{\hopto}{ \mathrel{\tikz[baseline=-0.75ex]{\draw[line width=0.1ex] (0ex,0ex) -- (1ex,0ex) -- ++(-0.5ex,0.5ex) ++(0ex,-1ex) -- ++(0.5ex,0.5ex);}}}

\newcommand*{\HS}[1]{\Omega_{#1}}
\newcommand*{\EA}[1]{\mathfrak{A}_{#1}}
\newcommand*{\CS}[1]{\Phi_{#1}}
\newcommand*{\FA}[1]{\mathcal{A}_{#1}}
\newcommand*{\FD}[1]{\mathcal{D}_{#1}}
\newcommand*{\FDprec}[1]{\FD{#1}^{\textup{prec}}}
\newcommand*{\FDprol}[1]{\FD{#1}^{\textup{prol}}}

\newcommand*{\Pow}[2][]{\mathcal{P}\mathopen #1 ( #2 \mathclose #1 )}
\newcommand*{\CovS}[1]{\phi^{(#1)}}
\newcommand*{\HistS}[1]{\gamma^{(#1)}}

\newcommand*{\Prev}{\vert_{-}}

\newcommand*{\PrecProl}[1]{\CS{#1}^{\textup{prec-prol}}}
\newcommand*{\Prec}[1]{\CS{#1}^{\textup{prec}}}
\newcommand*{\Prol}[1]{\CS{#1}^{\textup{prol}}}
\newcommand*{\Classical}[1]{\CS{#1}^{\textup{next Cl}}}
\newcommand*{\ClassicalS}[1]{\CS{#1}^{\textup{Cl}}}
\newcommand*{\Basic}[1]{\CS{#1}^{\textup{next Basic}}}
\newcommand*{\BasicS}[1]{\CS{#1}^{\textup{Basic}}}

\newcommand*{\GlobalS}[1]{\CS{#1}^{\textup{Global}}}
\newcommand*{\Extdef}{T^+}
\newcommand*{\Ext}[2][]{\Extdef #1 ( #2 \mathclose #1 )}

\begin{document}
\title{Evolving Realities for Quantum Measure Theory}
\author{Henry Wilkes\thanks{Theoretical Physics Group, Blackett Laboratory, Imperial College, London. SW7 2AZ. UK. E-mail: \nolinkurl{hw2011@ic.ac.uk}}}
\maketitle

\begin{abstract}
We introduce and explore Rafael Sorkin's \textit{evolving co-event scheme}: a theoretical framework for determining completely which events do and do not happen in evolving quantum, or indeed classical, systems. The theory is observer-independent and constructed from discrete histories, making the framework a potential setting for discrete quantum cosmology and quantum gravity, as well as ordinary discrete quantum systems. The foundation of this theory is Quantum Measure Theory, which generalises (classical) measure theory to allow for quantum interference between alternative histories; and its co-event interpretation, which describes whether events can or can not occur, and in what combination, given a system and a quantum measure. In contrast to previous co-event schemes, the evolving co-event scheme is applied in stages, in the stochastic sense, without any dependence on later stages, making it manifestly compatible with an evolving block view. It is shown that the co-event realities produced by the basic evolving scheme do not depend on the inclusion or exclusion of zero measure histories in the history space, which follows non-trivially from the basic rules of the scheme. It is also shown that this evolving co-event scheme will reduce to producing classical realities when it is applied to classical systems.
\end{abstract}

\section{Introduction}

Quantum Measure Theory (QMT) \cite{1994_sum_rules_Rafael_Sorkin,1994_restricted_preclusion_Rafael_Sorkin,2007_co-events_Rafael_Sorkin}, at its basis, takes probability measure theory and weakly extends it to accommodate quantum interference. Whilst the usual ``Hilbert space, operators and wavefunctions'' formulation of quantum mechanics will predict probabilities, they are restricted to ``operator at some time''-based events, and the theory is thus unable to answer inherently spacetime questions and lacks a description without observers. Such restrictions may not hinder quantum mechanic's application in non-relativistic and non-gravitational scenarios, but even here we still arguably lack a clear understanding of what our quantum systems are actually doing. Moreover, the measurement/collapse mechanic's classical/quantum split, and the special parameter status of time in this usual formulation are both obstacles to formulating a theory of quantum gravity or cosmology.

In contrast, QMT, which was constructed with the spacetime model of causal sets in mind, uses spacetime objects -- \textit{histories} -- as the basis of its theory, and does not feature any observer dependence or any collapse mechanic. The use of histories also allows us to treat quantum and classical objects similarly, keeping the theory general and applicable to many systems. What a history exactly is depends on the system being studied, but in general it will be a full (spacetime) description of a system's evolution. Sections II C-D in reference \cite{1993_spacetime_QM_lectures_Jame_Hartle}, by James Hartle, give a good demonstration of how an ordinary quantum system with a Hilbert space can give rise to histories using projection operators, and Section IV shows how this can be generalised to extend quantum mechanics to any system, including a closed universe, using histories and without constructing a Hilbert space. For example, if we wanted to describe two particles travelling on a fixed background spacetime, each history would be a pair of trajectories through the given spacetime. Moreover, we are particularly interested in making statements about \textit{events}, which are simply sets of histories. For the two particle system, one possible event would be the collection of all histories that satisfy ``at least one particle passes through the spacetime region $R$''.

In addition, whilst Hilbert space quantum mechanics uses the Hamiltonian and collapse for its dynamics, in QMT we use the \textit{quantum measure}, which measures the sum of quantum interferences between pairs of histories in an event. Again, \cite{1993_spacetime_QM_lectures_Jame_Hartle} shows how the evolution probabilities of a Hilbert space quantum system can be generalised to define a quantum measure\footnote{In fact, the reference derives a decoherence function, but the quantum measure is the diagonal of this binary function.}. Therefore, for ordinary quantum systems, when we restrict ourselves to events that correspond to sequenced measurement outcomes, the quantum measure will return the usual probabilities of quantum mechanics. However, if we then consider all possible events, including those that ordinary quantum mechanics is silent about, the quantum measure no longer obeys the necessary sum rules for (classical) measures because there now exists quantum interference between events. As such, the quantum measure can not, in general, be interpreted as a likelihood. Note that the Decoherent Histories approach to quantum mechanics uses these same basic objects and observations, but from here on QMT differs from this other approach.

Traditionally, QMT is formally introduced as a triple $(\Omega,\mathfrak{A},\mu)$ in analogue to (classical) measure theory, where $\Omega$ is the \textit{history space}; $\mathfrak{A}$ is the \textit{event algebra}: a Boolean subset of $\Pow{\Omega}$ (power set of $\Omega$); and $\mu:\mathfrak{A} \to \mathbb{R}$ is the quantum measure. This is the same triple that would be used in (classical) measure theory, except the event algebra is not necessarily a full $\sigma$-algebra, and the quantum measure does not obey the necessary sum rule for a (classical) measure. Since the quantum measure is not, in general, a probability measure, in QMT we seek a different way to interpret what the quantum measure tells us about reality. In particular, we only assume the statement ``events with zero quantum measure do not happen'', as we would in interpreting a probability measure, and from here we would like to know what versions of reality are possible or not.

To do so, we represent a possible reality by a \textit{co-event}, which is a map from the event algebra $\mathfrak{A}$ to $\Z{2}=\set{ 0,1 }$. An event occurs in the reality given by a co-event if the co-event maps it to $1$, and an event does not occur if it is mapped to $0$. A small set of these maps correspond to classical versions of reality, where classical logic is obeyed. Equivalently, these correspond to a reality where one history is \textit{the} history that happens, and whether some event happens or not depends only on whether it contains this history. For example, for a ball choosing to pass through one of three slits, one classical reality would be given by ``the ball goes through the middle slit'', and from here we could determine that the event ``the ball goes through the left or middle slit'' also happens, whilst the event ``the ball goes through the left or right slit'' does not happen. Whilst these classical versions of reality are certainly desirable, we find in QMT numerous examples where such classical realities are incompatible with quantum interference, such as the three-slit experiment in \cite{1994_restricted_preclusion_Rafael_Sorkin}, a Peres-Kochen-Specker set up in \cite{2008_kochen_specker_Fay_Dowker_and_Yousef_Ghazi-Tabatabai}, or the GHZ experiment in \cite{unpublished_ghz_experiment_Fay_Dowker}. Thus, in QMT we expect the co-events to have a more general structure that, precisely, expresses the counter-intuitive reality of quantum mechanics. The challenge of QMT has been deciding on a set of rules that use the quantum measure to derive some set of co-events that correspond to the realities that are allowed. To emphasise, the co-events are not intended to represent the limits of our `knowledge' of quantum systems, but are meant to represent the actual realities that are expressed. There have been a number of different such \textit{co-event schemes} \cite{2009_quantum_measure_schemes_Yousef_Ghazi-Tabatabai} developed under the general constraint that they must return to classical realities in classical scenarios (either by using a classical measure or through coarse graining) and that there must always be at least one allowed co-event, otherwise no reality could exist (which is far worse than a reality where nothing exists). One of the schemes that was shown to pass both criteria for finite $\Omega$ was the multiplicative scheme \cite{2007_anhomomorphic_logic_Rafael_Sorkin}, which also had a number of other appealing properties and uses \cite{2008_kochen_specker_Fay_Dowker_and_Yousef_Ghazi-Tabatabai,2012_logic_to_quantum_Rafael_Sorkin,2017_modus_ponens_Fay_Dowker_and_Kate_Clements_and_Petros_Wallden,2013_3_site_energy_Rafael_Sorkin}.

However, for systems that keep evolving, the above history space $\Omega$ would consist of histories that describe the system across an eternity, or some final time of the universe. Whilst this is not necessarily a mathematical issue, it gives us an inherently global perspective. Whilst we can classically go from a global description of reality (what happened everywhere from zero to infinity) to a time-finite one with the same classical properties (what happened up until this stage) this is not necessarily the case for a general co-event. Moreover, an investigation of the $n$-site hopper -- a finite quantum system that evolves in discrete time steps indefinitely -- led to the conclusion that any multiplicative co-event (one of the criteria for the multiplicative scheme) would deny every time-finite event \cite{2010_extending_quantum_measure_Fay_Dowker_and_Steven_Johnston_and_Sumati_Surya,2017_n-site_nirvana_Fay_Dowker_and_Vojtech_Havlicek_and_Cyprian_Lewandowski_and_Henry_Wilkes}. And so if we were to ask such a co-event what happened at any finite time we would leave empty-handed. Moreover, initial attempts to modify the multiplicative scheme to remove these global properties were prone to failure or trivialities, although there have been some other suggested modifications since \cite{2018_evolving_multiplicative_scheme_Rutvij_Bhavsar}.

This prompted the idea of an \textit{evolving} co-event scheme, which was suggested by Rafael Sorkin\footnote{Presented to the author during a discussion at Raman Research Institute in January 2016.}, which would describe reality up until some finite stage. To accommodate this we take a more stochastic approach: the system under consideration evolves in discrete \textit{stages} $t=0,1,\dotsc$ and for each stage we can define the quad $(\HS{t},\Prev,\EA{t},\mu_t)$, where the new object $\Prev$ is a function that relates histories in $\HS{t}$ to their past. This shifts the focus from global histories to histories that describe the system at some stage $t$, which also has the aesthetic benefit of making the term `history' more appropriate. Therefore, for each stage we can define a co-event from $\EA{t}$ to $\Z{2}$, which would fully describe what has happened up until that stage, but would not say anything about what happens at later stages. Now, an evolving scheme (with specifics given later) would construct the set of allowed co-events at some given stage $t$ using the allowed co-events from stage $t-1$ and the quantum measure $\mu_t$. Moreover, the evolving schemes in this paper will require that the new co-events do not dispute the past. This brings the evolving co-event scheme in line with the philosophical notion that the past remains fixed and that reality is determined in the present without reference to the future, which was noticeably absent in previous co-event schemes. This makes the scheme manifestly compatible with the growing block view, rather than being necessarily eternalist or presentist. Note that the growing block view is often adopted within causal set theory, without conflicting with the notion of general covariance \cite{2014_causal_sets_birth_Fay_Dowker,2010_causal_sets_non_existent_future_Rafael_Sorkin}.

We will begin with a number of definitions and symbolic notations in Section \ref{definitions}, which will finish with a summary of the concepts. In Section \ref{evolving_schemes} we will discuss what a classical co-event model would look like, observing key features of such a model and then generalising it to produce a basic evolving scheme for all (including both quantum and classical) stochastic systems. Section \ref{claims}'s focus is Theorem \ref{zero_hist_not_in_min_prec_prol_supp}, which claims that any co-event produced by the basic evolving scheme will not depend on histories with zero quantum measure, but the proof relies on a number of other claims, which may be of general interest and could be used for future proofs. We finish this section with a proof that the basic evolving scheme reduces to the classical evolving scheme. Finally, in Section \ref{more_evolving_schemes} we will explore how these claims carry over to more restrictive schemes, before finishing with a few unanswered questions for evolving schemes in Section \ref{discussion}.

\section{Definitions} \label{definitions}

Before we can properly introduce the basic evolving scheme, we will need a number of definitions and concepts. Now, whilst the objects defined in this section do not differentiate between classical or quantum systems, they are only defined here for systems that evolve in discrete stages that are ordered and can be correspondingly labelled by the natural numbers $t \in \mathbb{N}$. We require that at each stage, for every previous history there must exist some extension of that history to the new stage, i.e. for every ``$X$'' there must exist an ``$X$ and then ...'' at the next stage. Moreover, every history must be an extension of a previous one, i.e. there can not be a history with no past, unless it is the origin. Finally, for this paper we will only work with systems where the initial history space and number of such extensions is finite (and therefore the history space at all stages is finite). Also, note that in this paper $\Z{n}$ is only intended to represent the set of $n$ integers $\set{ 0,1, \dotsc,n-1 }$.

In order to aid with these definitions, we will occasionally refer to three different example systems:
\begin{itemize}
	\item Quantum $n$-site Hopper. This system consists of a single quantum particle -- the hopper -- on a discrete ring of $n$ sites \cite{2012_towards_a_fundamental_Rafael_Sorkin}. At each \textit{time} $\tau=0,1,\dotsc$ the hopper can change sites with transition amplitude $U_{ij}$, where $i$ labels the original site and $j$ the new one. In previous literature $U$ has a particular form, but for the purpose of this paper can be considered any unitary matrix that would in the state-vector approach to quantum mechanics correspond to the evolution matrix.
	\item Classical Random Walker in a Box. This system consists of a single classical random walker in a discrete $1d$ box, with only three spatial positions. At each time $\tau=0,1,\dotsc$ the walker can either choose to move left or right in the box with equal probability, but if it chooses right when it is already at the right-most site it will rebound onto the same site, and similarly for the left-most site. We have introduced this system to demonstrate that the work in this paper can be applied to classical systems and can return the expected classical results.
	\item Growing Labelled Causal Set. This system consists of a single labelled causal set that increases in size. A causal set is a locally-finite poset model for spacetime, where the order relation $a \prec b$ corresponds to the geometrical notion of $a$ is in the past of $b$. A labelled causal set is a poset with natural numbers attached to the elements, which is like putting coordinates on a spacetime. These growing causal sets are considered in Classical Sequential Growth (CSG) models \cite{1999_classical_sequential_growth_David_Rideout_and_Rafael_Sorkin}, where at each stage a single new element is added to the causal set above (in the partially ordered sense) or unrelated to existing elements. Moreover, such a model could also be incorporated into an eventual quantum theory of causal sets. The details of causal sets or CSG models are not important to understand this paper, we only mention it here for those interested in quantum gravity, and to highlight the general applicability of our theory.
\end{itemize}

\subsection{Histories and Events}

Since QMT is a realist theory, we want to describe the full evolution of our system up to some stage $t$. Such a description is referred to as a history of the system. For our $n$-site hopper, a history at stage $t$ would describe where the hopper is at each time $\tau$ from $\tau=0$ up until $\tau=t$. Note the subtle difference between ``time'' and ``stage'', the former being a parameter of the hopper's path, and the latter being a measure of how many steps our system has progressed by. For example, one such history at stage $2$ would be the path $ 1 \hopto 4 \hopto 3 $, see Figure \ref{histories_figure} for a visualisation of this history. The set of all such histories at stage $t$ is the history space.
\begin{definition}[History Space]
	The history space at stage $t$, denoted $\HS{t}$, for a given system is the finite set of all histories that describe the system's evolution (including its past) at stage $t$.
\end{definition}
\newlength{\LW}
\setlength{\LW}{0.002\textwidth}
\newlength{\La}
\setlength{\La}{0.08\textwidth}
\begin{figure*}[t!]
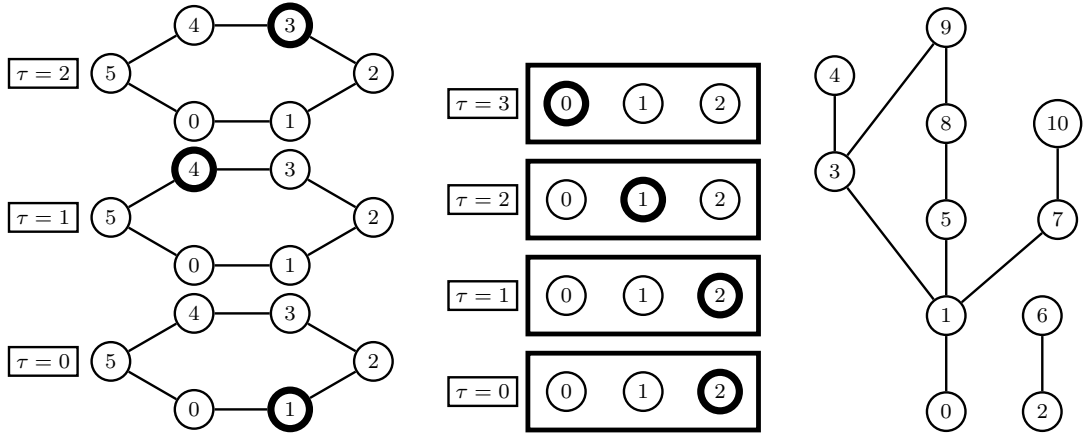

	\centering
	\footnotesize
	\newcommand*\addsite[3]{\node[circle,draw,line width=\LW] (site_#3) at (#1) {$#2$} }
	\newcommand*\addboldsite[3]{\node[circle,draw,line width=3\LW] (site_#3) at (#1) {$#2$} }
	\newcommand*\addtime[2]{\node[left,rectangle,draw,line width=\LW] at (#1) {$\tau=#2$} }
	\tikz{
		\addsite{0,0}{0}{00};
		\addboldsite{\La,0}{1}{01};
		\addsite{1.866025404\La,0.5\La}{2}{02};
		\addsite{\La,\La}{3}{03};
		\addsite{0,\La}{4}{04};
		\addsite{-0.866025404\La,0.5\La}{5}{05};
		\addtime{-1.2\La,0.5\La}{0} ;
		\draw[line width=\LW] (site_00) -- (site_01) -- (site_02) -- (site_03) -- (site_04) -- (site_05) -- (site_00);
		\addsite{0,1.5\La}{0}{16};
		\addsite{\La,1.5\La}{1}{17};
		\addsite{1.866025404\La,2\La}{2}{18};
		\addsite{\La,2.5\La}{3}{19};
		\addboldsite{0,2.5\La}{4}{110};
		\addsite{-0.866025404\La,2\La}{5}{111};
		\addtime{-1.2\La,2\La}{1} ;
		\draw[line width=\LW] (site_16) -- (site_17) -- (site_18) -- (site_19) -- (site_110) -- (site_111) -- (site_16);
		\addsite{0,3\La}{0}{212};
		\addsite{\La,3\La}{1}{213};
		\addsite{1.866025404\La,3.5\La}{2}{214};
		\addboldsite{\La,4\La}{3}{215};
		\addsite{0,4\La}{4}{216};
		\addsite{-0.866025404\La,3.5\La}{5}{217};
		\addtime{-1.2\La,3.5\La}{2} ;
		\draw[line width=\LW] (site_212) -- (site_213) -- (site_214) -- (site_215) -- (site_216) -- (site_217) -- (site_212);
	}
	\hspace{0.4\La}
	\tikz{
		\draw[line width=2\LW] (-0.4\La,-0.4\La) rectangle (2\La,0.4\La) ;
		\addtime{-0.5\La,0}{0};
		\addsite{0,0}{0}{00};
		\addsite{0.8\La,0}{1}{01};
		\addboldsite{1.6\La,0}{2}{02};
		\draw[line width=2\LW] (-0.4\La,0.6\La) rectangle (2\La,1.4\La) ;
		\addtime{-0.5\La,\La}{1};
		\addsite{0,\La}{0}{10};
		\addsite{0.8\La,\La}{1}{11};
		\addboldsite{1.6\La,\La}{2}{12};
		\draw[line width=2\LW] (-0.4\La,1.6\La) rectangle (2\La,2.4\La) ;
		\addtime{-0.5\La,2\La}{2};
		\addsite{0,2\La}{0}{20};
		\addboldsite{0.8\La,2\La}{1}{21};
		\addsite{1.6\La,2\La}{2}{22};
		\draw[line width=2\LW] (-0.4\La,2.6\La) rectangle (2\La,3.4\La) ;
		\addtime{-0.5\La,3\La}{3};
		\addboldsite{0,3\La}{0}{30};
		\addsite{0.8\La,3\La}{1}{31};
		\addsite{1.6\La,3\La}{2}{32};
	}
	\hspace{0.4\La}
	\tikz{
		\addsite{0,0}{0}{0};
		\addsite{0,\La}{1}{1};
		\addsite{\La,0}{2}{2};
		\addsite{-1.16\La,2.5\La}{3}{3};
		\addsite{-1.16\La,3.5\La}{4}{4};
		\addsite{0,2\La}{5}{5};
		\addsite{\La,\La}{6}{6};
		\addsite{0,3\La}{8}{8};
		\addsite{1.16\La,2\La}{7}{7};
		\addsite{0,4\La}{9}{9};
		\addsite{1.16\La,3\La}{10}{10};
		\draw[line width=\LW] (site_0) -- (site_1) -- (site_3) -- (site_4)
		(site_3) -- (site_9)
		(site_1) -- (site_7) -- (site_10)
		(site_1) -- (site_5) -- (site_8) -- (site_9)
		(site_2) -- (site_6);
	}
	\caption{\label{histories_figure}A visualisation of histories for three different systems. The left history is a member of $\HS{2}$ for the $6$-site hopper, and represents the path $1 \hopto 4 \hopto 3$. The middle history is a member of $\HS{3}$ for the random walker in a box, and represents the path $2 \hopto 2 \hopto 1 \hopto 0$, where between time $\tau=0$ and $\tau=1$ the walker bounced off the right-most wall. The right history is a member of $\HS{11}$ for the growing labelled causal set, visualised here as a Hasse diagram: two elements $a,b$, represented as nodes, are partially ordered as $a \prec b$ if there exists a path from $a$ to $b$ that only travels upwards. The numbers inside the nodes are the labels. For example, the element labelled $3$ is above $1$ and $0$, below $4$ and $9$, and not related to any other elements.}
\end{figure*}

We can formally define our $n$-site hopper history space at stage $t$ to be
\begin{equation}
	\setdef{\gamma}{\gamma:\Z{t+1} \to \Z{n}} \,,
\end{equation}
where $\gamma$ is a path and $\gamma(\tau)=i$ means the hopper is at site $i$ at time $\tau$. For the random walker, the system's evolution is very similar to the $3$-site hopper's so we might similarly choose the history space at stage $t$ to be
\begin{equation}
	\setdef{\gamma}{\gamma:\Z{t+1} \to \Z{3}} \,,
\end{equation}
where $\gamma(\tau)$ refers to the position of the walker at time $\tau$, see Figure \ref{histories_figure} for an example. However, we have included histories that would not be allowed in our description of the walker system, for example paths that contain the transition $0 \hopto 2$, which can not occur by just choosing left or right. To account for this we could make the history space smaller by excluding these histories. However, as in classical theories, we would hope that we can simply set the measure of such disallowed histories to zero, thus essentially excluding them from the theory. We shall see later than the basic evolving scheme will indeed give this desirable result.

Finally, for a theory of growing labelled causal sets the history space at stage $t$ would simply be the set of all labelled causal set that contain $t$ elements, see Figure \ref{histories_figure} for an example. Of note here, unlike the other two examples, the stage $t$ is \textit{not} related to ``what state is our system in at times $\tau \leq t$'', but instead to how big our spacetime is.

Now, returning to the $n$-site hopper, if our system's evolution is described by $1 \hopto 4 \hopto 3$ at stage $2$, then when we restrict ourselves to describing the system at stage $1$, then the corresponding history is $1 \hopto 4$. This introduces the concept of a restriction map from later to earlier stages.
\begin{definition}[Restriction Map]
	The restriction map at stage $t>0$ is a map 
	\begin{equation}
		\begin{aligned}
			\Prev:\HS{t} & \to \HS{t-1}
			\\ \gamma & \mapsto \gamma \Prev
		\end{aligned}
	\end{equation}
	that is surjective (onto), and physically relates each history to one previous history, which is its past.
\end{definition}

For systems like the $n$-site hopper or the random walker, where each stage corresponds to an absolute final time, the restriction map $\Prev$ would just remove from the history the last thing the system did. So in this case $1 \hopto 4 \hopto 3 \Prev = 1 \hopto 4$, and more generally, the restriction would be given by
\begin{equation}
	\gamma\Prev(\tau) = \gamma(\tau) \quad \forall \tau \in \Z{t} \, ,
\end{equation}
which is the same path as $\gamma$, but is only defined up until time $\tau=t-1$, so is missing the last step. In terms of the diagrams in Figure \ref{histories_figure}, it would correspond to removing the top sections.

For a theory of growing labelled causal sets, since the stage $t$ corresponds to the number of elements, the restriction map $\Prev$ acting on a growing labelled causal set would be the same labelled causal set minus the last element that was added and any corresponding relations, so for the labelled causal set in Figure \ref{histories_figure} we would remove the element labelled $10$ and its relations to $7$, $1$ and $0$.

Moreover, if $\gamma'$ restricts to $\gamma$, then it is natural to think of $\gamma'$ as an \textit{extension} of $\gamma$. The surjective property of the restriction guarantees that every history will have at least one extension to the next stage (no matter what a system has done it can always evolve further, provided we run it for another stage). Note that for the $n$-site hopper there are $n$ extensions for every $\gamma \in \HS{t}$ to the next stage, corresponding to the $n$ different sites the hopper could switch to. However, for the labelled causal set, and more generally, the number of extensions may vary with each $\gamma$.

Now, in QMT we are interested in answering what events do and do not happen. For example, we might be interested in whether our hopper is at site $i$ at time $\tau$. Mathematically, such an event is represented at stage $t$ by a collection of histories that all place the hopper at site $i$ and time $\tau$, i.e. the set
\begin{equation}
	E = \setdef{\gamma \in \HS{t}}{\gamma(\tau)=i}\,. \label{example_event}
\end{equation}
We can take this further and make any specifications about our system's evolution, and represent this mathematically as a collection of histories at stage $t$ that satisfy these conditions. Note that this is provided the conditions do not pertain to any details about the system's evolution that go beyond what is possible at the current stage. For example, if $\tau=3$ in the example event, then we can not represent this at stages $t=0$, $1$ or $2$ because there has been no such time yet, but once we reach stage $3$ or later we can begin to represent the event. Moreover, a statement like ``the hopper is at site $i$'' is \textit{not} an event because it has no anchor in time, which means it can not be constructed from histories. In general, the full set of events at stage $t$ is the event algebra.

\begin{definition}[Event Algebra]
	The event algebra at stage $t$, denoted $\EA{t}$, is the power set of $\HS{t}$\footnote{In general the event algebra in QMT is defined as the domain of the quantum measure, and is only required to be a Boolean subset of the power set of the history space, but since we are only dealing with finite history spaces in this paper we always choose the event algebra to be the full power set. For some issues related to infinite history spaces see \cite{2010_extending_quantum_measure_Fay_Dowker_and_Steven_Johnston_and_Sumati_Surya}.}. The elements of $\EA{t}$ are called events.
\end{definition}

Most events in $\EA{t}$ are far more complex to specify than in our previous example, but all are relevant to our theory.	Moreover, we can associate the event algebra with a Boolean ring by defining the binary operations
\begin{listeqn}
	A \cdot B := A \cap B \label{event_and}
	\AND A + B := (A \cup B) \setminus (A \cap B) \,, \label{event_xor}
\end{listeqn}
and the elements
\begin{listeqn}
	0 := \set{\ }
	\AND 1 := \Omega_t \,.
\end{listeqn}
Note that \eqref{event_and} is the event ``$A$ and $B$'' and \eqref{event_xor} is the event ``$A$ or $B$ but not both''. In addition, $A+A=0$ and $(1+A)$ is the complement of $A$.

Apart from the empty event $0$, the smallest events in our theory at stage $t$ are those that only contain one history. We call such events \textit{single-history events}. In this report, for convenience we will use $\gamma$ to represent both the history itself, and the single-history event $\set{\gamma} \in \EA{t}$, where the context will differentiate between the two meanings. In particular, we will often write $E+\gamma$ for some event $E \in \EA{t}$, which will mean
\begin{equation}
	E + \gamma = \begin{cases} E \setminus \set{\gamma} \quad & \text{if } \gamma \in E
		\\ E \cup \set{\gamma} \quad & \text{if } \gamma \not\in E \,.
	\end{cases}
\end{equation}

Now, we have discussed how we can mathematically represent physical conditions on a system's evolution using events from $\EA{t}$. The only restriction is that the stage $t$ must be late enough to capture all the details. But we could choose any other stage $t'>t$ and do the same. To switch from an event $E$ at stage $t-1$ to the same physically equivalent event at stage $t$, we take all the histories in $E$ and find all of their extensions to the next stage $t$. This defines the extension function.
\begin{definition}[Extension of an Event]
	The extension function for any $t>0$ is given by
	\begin{equation}
		\begin{aligned}
			\Extdef:\EA{t-1} & \to \EA{t}
			\\  E & \mapsto \Ext{E}:=\setdef{\gamma \in \HS{t}}{\gamma\Prev \in E} \,.
		\end{aligned}
	\end{equation}
\end{definition}
For our example event in \eqref{example_event} the extension is simply given by
\begin{equation}
	\Ext{E}=\setdef{\gamma \in \HS{t+1}}{\gamma(\tau)=i} \,.
\end{equation}
More generally, this physical equivalence is captured by the equivalence
\begin{equation}
	\gamma \in \Ext{E} \iff \gamma\Prev \in E  \,.\label{ext_E_is_phys_equiv_to_E}
\end{equation}
In fact, mathematically, $\Extdef$ is the \textit{preimage} for the function $\Prev$. It follows that for all events $A,B \in \EA{t}$
\begin{listeqn}
	\Ext{A+B}=\Ext{A}+\Ext{B} 
	\COM \Ext{A \cdot B} = \Ext{A} \cdot \Ext{B} 
	\COM \Ext{1} = 1
	\AND \Ext{0} = 0 \,.
\end{listeqn}

Now, the corresponding reverse operation from $\EA{t}$ to $\EA{t-1}$ would be to take all the histories in the event and restrict them to the previous stage. This defines the restriction map for events.
\begin{definition}[Restriction of an Event]
	Using the restriction map $\Prev:\HS{t} \to \HS{t-1}$ we define a corresponding restriction map on the event algebra for all $t>0$
	\begin{equation}
		\begin{aligned}
			\Prev : \EA{t} & \to \EA{t-1}
			\\ E & \mapsto	E \Prev = \bigcup \limits_{\gamma \in E} \gamma\Prev \, .
		\end{aligned}
	\end{equation}
\end{definition}
For our example event in \eqref{example_event}, if $\tau<t$ then
\begin{equation}
	E \Prev = \setdef{\gamma \in \HS{t-1}}{\gamma(\tau)=i} \,.
\end{equation}
However, this operation will in general remove any information about what our system did at the last stage, so if $\tau=t$ then in this case it restricts to
\begin{equation}
	E \Prev = \HS{t-1} \,.
\end{equation}

Again, mathematically this is the \textit{image} of the event $E$ under the function $\Prev$. It follows that for all events $A,B \in \EA{t}$
\begin{listeqn}
	\Ext{A\Prev} \supseteq A
	\AND A\Prev \cdot B\Prev \supseteq (A \cdot B) \Prev \,.
\end{listeqn}

These restriction and extension maps are important to understand for the rest of the paper, so they have been summarised graphically in Figure \ref{stage_operations_summary}.

\begin{figure*}[t!]
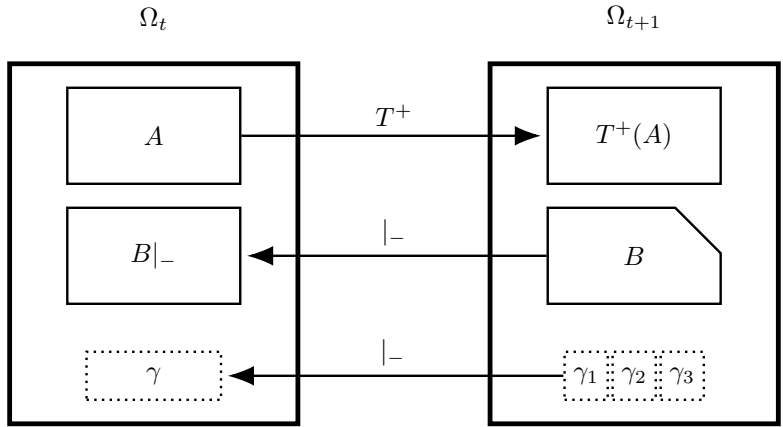

	\centering
	\tikz{
		\draw[line width=2\LW] (0,0) rectangle (3\La,3.75\La);
		\node [above] at (1.5\La,4\La) {$\HS{t}$};
		\node[rectangle,draw,line width=\LW, minimum width=1.8\La, minimum height=\La] (event_a) at (1.5\La,3\La) {$A$};
		\node[rectangle,draw,line width=\LW, minimum width=1.8\La, minimum height=\La] (event_b_prev) at (1.5\La,1.75\La) {$B\Prev$};
		\node[rectangle,draw,line width=\LW, minimum width=1.4\La, minimum height=0.5\La, dotted] (hist) at (1.5\La,0.5\La) {$\gamma$};
		\draw[line width=2\LW] (5\La,0) rectangle (8\La,3.75\La);
		\node [above] at (6.5\La,4\La) {$\HS{t+1}$};
		\node[rectangle ,draw,line width=\LW, minimum width=1.8\La, minimum height=\La] (event_a_ext) at (6.5\La,3\La) {$\Ext{A}$};
		\node[shape=chamfered rectangle, chamfered rectangle xsep=0.5\La, inner sep=0.3\La, chamfered rectangle corners=north east,draw,line width=\LW, minimum width=1.8\La, minimum height=\La] (event_b) at (6.5\La,1.75\La) {$B$};
		\node[rectangle,draw,line width=\LW, minimum width=0.4\La, minimum height=0.5\La,dotted] (hist_1) at (6\La,0.5\La) {$\gamma_1$};
		\node[rectangle,draw,line width=\LW, minimum width=0.4\La, minimum height=0.5\La,dotted] (hist_2) at (6.5\La,0.5\La) {$\gamma_2$};
		\node[rectangle,draw,line width=\LW, minimum width=0.4\La, minimum height=0.5\La,dotted] (hist_3) at (7\La,0.5\La) {$\gamma_3$};
		\draw[-{Latex[length=10, sep=2pt]}, line width=\LW] (event_a) to node [auto] {$\Extdef$} (event_a_ext);
		\draw[-{Latex[length=10, sep=2pt]}, line width=\LW] (event_b) to node [auto,swap] {$\Prev$} (event_b_prev);
		\draw[-{Latex[length=10, sep=2pt]}, line width=\LW] (hist_1) to node [auto,swap] {$\Prev$} (hist);
	}
	\caption{\label{stage_operations_summary}A graphical summary of the operations on events between the history space at stage $t$ and the history space at stage $t+1$. Events are represented by bold lines and single histories are dashed lines. Graphically, we can think of $\HS{t+1}$ as a copy of $\HS{t}$ with extra details about what the system does next. Therefore, we take the areas in $\HS{t}$ and copy them into $\HS{t+1}$ and divide them into $n$-segments for each history, where $n$ is the number of extensions the history has. Here the event $A$ is mapped to $\Ext{A}$, which contains the same details. The event $B$ is mapped back to a subset of $\HS{t}$ by $\Prev$, which contains less details. We depicted a corner missing in $B$ to emphasise that $B$ is always a subset of $\Ext{B\Prev}$. The history $\gamma$ here has three extensions: $\gamma_1$, $\gamma_2$ and $\gamma_3$, that all map back to $\gamma$ under $\Prev$.}
\end{figure*}
\subsection{Quantum Measure}

Now that we have the kinematic objects of our theory, we require some predictions of their expression. In histories approaches to quantum mechanics, this is given by the decoherence function, which measures the quantum interference between events by summing the interference between the histories in those events. In particular, in QMT we are interested in the diagonal of the decoherence function, which we call the quantum measure.

For this paper, we will refer to the object that measures the interference between pairs of histories as the decoherence matrix.
\begin{definition}[Decoherence Matrix]
	The decoherence matrix at stage $t$, denoted $d_t$, is a physically motivated matrix $d_t:\HS{t} \times \HS{t} \to \mathbb{C}$ that describes the quantum interference between the histories; is Hermitian with real and non-negative eigenvalues; satisfies the consistency condition
	\begin{equation}
		d_t(\gamma_1,\gamma_2) = \sum \limits_{\substack{\gamma'_1 \in \HS{t'} \\ \gamma'_1 \Prev = \gamma_1}} \sum \limits_{\substack{\gamma'_2 \in \HS{t'} \\ \gamma'_2 \Prev = \gamma_2}} d_{t'}(\gamma'_1, \gamma'_2) \,; \label{decoherence_consistency}
	\end{equation}
	and is normalised as
	\begin{equation}
		\sum \limits_{\gamma_1 \in \HS{t}} \sum \limits_{\gamma_2 \in \HS{t}} d_t(\gamma_1,\gamma_2) = 1\,.
	\end{equation}
\end{definition}
Both the consistency condition and the normalisation conditions are in place to ensure that the decoherence function, as given below, satisfies
\begin{equation}
	D_t(A,B)=D_{t+1}\bigl( \Ext{E} \bigr) \label{qmeasure_consistency}
\end{equation}
and
\begin{equation}
	D_t(\HS{t},\HS{t})=1
\end{equation}
respectively, with the former ensuring that physically equivalent events are always assigned the same quantum measure, and the latter ensuring that the quantum measure is a normalised probability measure for classical systems.

We can now give the definitions for the decoherence function and the quantum measure.
\begin{definition}[Decoherence Function]
	The decoherence function at stage $t$, denoted $D_t$, is
	\begin{equation}
		\begin{aligned}
			D_t:\EA{t} \times \EA{t} & \to \mathbb{C}
			\\ (A,B) & \mapsto D_t(A,B):= \sum \limits_{\gamma_1 \in A} \sum \limits_{\gamma_2 \in B} d_t(\gamma_1,\gamma_2)  \,.
		\end{aligned}
	\end{equation}
\end{definition}
\begin{definition}[Quantum Measure]
	The quantum measure at stage $t$, denoted $\mu_t$, is given by the diagonal of the decoherence functional
	\begin{equation}
		\begin{aligned}
			\mu_t:\EA{t} & \to \mathbb{R}
			\\ E & \mapsto \mu_t(E) := D_t(E,E) \,.
		\end{aligned}
	\end{equation}
\end{definition}
For classical systems, like the random walker in a box, the decoherence matrix is always diagonal. Specifically, for the random walker
\begin{equation}
	d_t(\gamma_1,\gamma_2) = \begin{cases}
		P(\gamma_1) \quad & \text{if }\gamma_1=\gamma_2
		\\ 0 \quad & \text{otherwise,}
	\end{cases}
\end{equation}
where $P(\gamma_1)$ is the probability for taking the path $\gamma_1$, which would be zero for any paths with $0 \hopto 2$ or $2 \hopto 0$ or $1 \hopto 1$ transitions, and would otherwise depend on the probability of starting at the site $\gamma_1(0)$ and an overall normalisation. This means that the decoherence functional for this random walker is given by
\begin{equation}
	D_t(A,B) = \begin{cases}
		\sum \limits_{\gamma \in A} P(\gamma) \quad & \text{if } A=B
		\\ 0 \quad & \text{otherwise.}
	\end{cases}
\end{equation}
Thus, the quantum measure for an event will just be the sum of probabilities for each history in the event, and will therefore act as a probability measure.

On the other hand, for quantum systems like the $n$-site hopper, the decoherence matrix will not be diagonal. In particular, for the $n$-site hopper the decoherence matrix is zero if the two histories $\gamma_1$ and $\gamma_2$ do not meet at time $\tau=t$. Otherwise, if $\gamma_1(t)=\gamma_2(t)$ then it is given by the complex inner product of their path amplitudes $a[\gamma]$, which is given by an initial amplitude $\psi_{\gamma(0)}$ and a product of $U_{ij}$ for each transition $i \hopto j$ made by the history:
\begin{equation}
	a[\gamma] = \psi_{\gamma(0)} \prod \limits_{\tau=0}^{t-1} U_{\gamma(\tau),\gamma(\tau+1)} \,.
\end{equation}
Thus, the decoherence functional is given by
\begin{equation}
	D_t(A,B) = \sum \limits_{\gamma_1 \in A} \sum \limits_{\gamma_2 \in B} a[\gamma_1] a[\gamma_2]^* \delta_{\gamma_1(t),\gamma_2(t)} \,.
\end{equation}
Note, for unitary systems like this, the delta term plays a key role is ensuring the consistency condition \eqref{decoherence_consistency} holds.

Now, if we choose the event $A=B=\text{``the hopper is at site $i$ at time $\tau$''}$ the decoherence functional will indeed return the ordinary quantum mechanics probability of a positive measurement at site $i$ at time $\tau$, given an initial wavefunction of $\psi$. However, this only covers a very small set of events in $\EA{t}$, and, generally, a decoherence functional that is not diagonal will produce a quantum measure that does \textit{not} obey the classical sum rule
\begin{equation}
	\mu_t(A+B)=\mu_t(A)+\mu_t(B)
\end{equation}
for disjoint events $A,B \in \EA{t}$. Therefore, as stated earlier, we can not interpret the quantum measure as a likelihood. However, the Hermitian nature and non-negative eigenvalues of $d_t$ guarantee that the quantum measure is real and non-negative for all events. Moreover, all quantum measures will still obey the weaker quantum sum rule \cite{1994_sum_rules_Rafael_Sorkin}
\begin{equation}
	\mu_t(A+B+C) = \mu_t(A+B)+\mu_t(A+C)+\mu_t(B+C) - \mu_t(A) - \mu_t(B) - \mu_t(C)
\end{equation}
for disjoint events $A,B,C \in \EA{t}$.

Whilst the quantum measure can not, in general, be interpreted as a measure of likelihood, in QMT we seek an alternative way to interpret the quantum measure. In classical theories with finite history spaces we view zero measure events as impossible ones. In QMT we often work from this single assumption, applied to any events with zero quantum measure. Therefore, we are often focussed on such events, which we call null events.
\begin{definition}[Null Event]
	An event $E\in \EA{t}$ is a null event if
	\begin{equation}
		\mu_t(E)=0 \,.
	\end{equation}
\end{definition}
The weaker conditions on the quantum measure allows systems to have many more null events than they could have classically because the interference between histories can lead to cancellations. Also, note that $\mu_t(0)=0$, so the empty set is always a null event.

For this paper, we will also be interested in null histories.
\begin{definition}[Null History]
	A history $\gamma \in \HS{t}$ is a null history if
	\begin{equation}
		d_t(\gamma,\gamma)=0 \,.
	\end{equation}
\end{definition}
Note that the condition for a null history is equivalent to the corresponding single history \textit{event} being a null event. In addition, we can follow the same method as in deriving \eqref{d_on_P_is_zero} later on, to show that
\begin{equation}
	d_t(\gamma,\gamma)=0 \implies d_t (\gamma,\gamma')=0 \quad \forall \gamma' \in \HS{t} \label{null_means_no_interference} \,.
\end{equation}
Therefore, null histories can not interfere with any other histories.

Finally, we note that it is often convenient to introduce the following vector space. 
\begin{definition}[Vector Space $V_t$]
	$V_t$, defined for each history space $\HS{t}$, is a $\vert \HS{t} \vert$ dimensional complex vector space. Each vector $v \in V_t$ associates a complex number to each history in $\HS{t}$, which we represent as a function $v:\HS{t} \to \mathbb{C}$. $V_t$ is spanned by the basis vectors $v_\gamma$, $\gamma \in \HS{t}$, defined as
	\begin{equation}
		v_\gamma(\gamma') := \begin{cases}
			1 \quad & \text{if } \gamma'=\gamma
			\\ 0 \quad & \text{otherwise.}
		\end{cases}
	\end{equation}
	We also define a corresponding vector for any event $E \in \EA{t}$, given by
	\begin{equation}
		v_E(\gamma) := \begin{cases}
			1 \quad & \text{if } \gamma \in E
			\\ 0 \quad & \text{otherwise.}
		\end{cases}
	\end{equation}
\end{definition}
Using the $v_\gamma$ as a basis, we can relate the decoherence matrix $d_t$ to a linear transformation $V_t \to V_t$, which we also label as $d_t$. Then the decoherence function is simply given by
\begin{equation}
	D_t(A,B) = {v_A}^\dagger d_t v_B \,.\label{decoherence_using_vectors}
\end{equation}

\subsection{Co-events}

The co-event approach to QMT aims to use the quantum measure to derive a set of co-events, which each represent a \textit{different} version of reality. A single co-event represents reality to the extent that it specifies completely which events happen and which events do not happen. Thus, a co-event is simply a map from the event algebra to $\Z{2}$, where an event $E$ is mapped to $0$ if it does not happen, and is mapped to $1$ if it does happen. We say that the co-event \textit{denies} $E$ in the first case, and \textit{affirms} $E$ in the second case. For evolving co-event schemes, there is a different event algebra at each stage $t$, and thus the co-event will describe the reality of our system up until stage $t$ only. The full set of all such co-events is the co-event space.
\begin{definition}[Co-event Space]
	The co-event space at stage $t$, denoted $\CS{t}$, is the set of all maps $\phi:\EA{t} \to \Z{2}$, where each map is referred to as a co-event.
\end{definition}
There are therefore $2^{\vert \EA{t} \vert}=2^{2^{\vert \HS{t} \vert}}$ different versions of realities at stage $t$, as far as a co-event representation is concerned. Similar to $\EA{t}$, We can give both the common co-domain $\Z{2}$ and $\CS{t}$ a Boolean ring structure. For $\Z{2}$, we define the usual binary operations for elements $a,b \in \Z{2}$ as
\begin{listeqn}
	a \cdot b := a \mathrel{\texttt{AND}} b
	\AND a + b := a \mathrel{\texttt{XOR}} b \,.
\end{listeqn}
For $\CS{t}$, we define the same operations as
\begin{listeqn}
	(\phi_1 \cdot \phi_2)(E) := \phi_1(E) \cdot \phi_2(E) 
	\AND (\phi_1 + \phi_2)(E) := \phi_1(E) + \phi_2(E) \,,
\end{listeqn}
and we take the elements $0,1 \in \CS{t}$ to be
\begin{listeqn}
	0(E) = 0 \quad \forall E \in \EA{t}
	\AND 1(E) = 1 \quad \forall E \in \EA{t}\,.
\end{listeqn}

Note that we have already restricted what a reality could be: the co-events are definite, not allowing for a middle state (either an event happens or it does not happen); and the co-events are complete functions, not allowing for an event to be not related to either `affirmed' or `denied', nor allowing an event to be related to both. In addition, in QMT we aim to describe our own reality by a single co-event, as opposed to an ensemble of co-events. In practice the theory will predict a number of possible co-events, one of which is the one that is actually expressed (or if you would rather, all of which are expressed, but each one is a distinct reality). These early restrictions are kept in place since we ultimately want a realist description that will, in the correct setting, return a classical realist description.

However, we have still allowed for counter-intuitive realities in our space. For example, there are plenty of co-events where even though the event $A$ is affirmed and the event $B$ is affirmed, the event ``$A$ and $B$''$=A \cdot B$ does not happen. This extra freedom will become necessary once we allow for quantum interference. However, there is still a subset of co-events in $\CS{t}$ that obey our classical intuition, and these are known a classical co-events.
\begin{definition}[Classical Co-event]
	For every history $\gamma \in \HS{t}$ there exists a corresponding classical co-event $\gamma^* \in \CS{t}$ given by
	\begin{equation}
		\gamma^*(E):= \begin{cases}
			1 & \quad \text{if } \gamma \in E
			\\ 0 & \quad \text{otherwise.}
		\end{cases}
	\end{equation}
\end{definition}
We can intuitively think of the history $\gamma$ that specifies the classical co-event $\gamma^*$ as \textit{the} history that happens, with every other event occurring iff it involves the history $\gamma$. For example, for the random walker at stage $2$, one possible classical co-event would be
\begin{equation}
	0 \hopto 0 \hopto 1  ^* \,.
\end{equation}
So we would say that ``the walker takes the path $0 \hopto 0 \hopto 1$'' in this reality, and moreover, the event ``the walker hit a wall before time $2$'' does happen and the event ``the walker reached site $2$ before time $2$'' does not happen in this reality.

These classical co-events obey the rules of classical logic, that is they all satisfy, for all $A,B \in \EA{t}$,
\begin{listeqn}
	 \gamma^*(A+B) =\gamma^*(A)+\gamma^*(B)
	\COM \gamma^*(A \cdot B) =\gamma^*(A) \cdot \gamma^*(B)
	\COM \gamma^*(0) =0
	\AND \gamma^*(1) =1 \,.
\end{listeqn}
In fact, these conditions show that classical co-events are homomorphisms from the Boolean $\EA{t}$ to the Boolean $\Z{2}$. Moreover, one can show that they are the full set of such homomorphisms.

Whilst the classical co-events are a special subset of $\CS{t}$, as we shall see in Section \ref{expansion_of_coevents}, all co-events can be expanded as a polynomial of classical co-events. For example, the non-classical co-event
\begin{equation}
	\phi={\gamma_1}^* + {\gamma_1}^* \cdot {\gamma_2}^* + {\gamma_3}^* \label{example_coevent}
\end{equation}
would evaluate $\gamma_1$ as
\begin{align}
	\phi(\gamma_1) & = 1 + 1 \cdot 0 + 0
	\\ & = 1
\end{align}
and $\gamma_1+\gamma_2$ as
\begin{align}
	\phi(\gamma_1+\gamma_2) & = 1 + 1 \cdot 1 + 0
	\\ & = 0 \,.
\end{align}
Each product of classical terms in the sum is known as a monomial co-event.
\begin{definition}[Monomial Co-event]
	For every event $E \in \EA{t}$ there exists a corresponding monomial co-event $E^* \in \CS{t}$ given by
	\begin{equation}
		E^* := \prod \limits_{\gamma \in E} \gamma^*
	\end{equation}
	with
	\begin{equation}
		0^* := 1 \,.
	\end{equation}
\end{definition}
Note that $E^*$ can be equivalently defined as
\begin{equation}
	E^*(A)=\begin{cases}
		1 & \quad \text{if } A \supseteq E
		\\ 0 & \quad \text{otherwise.}
	\end{cases} \label{E_star_of_A}
\end{equation}

Now, we expect to use the quantum measure to restrict our set of co-events, and thus restrict what versions of reality our theory predicts are possible. In particular, as we mentioned earlier, we hope to use the principle that null events can not happen. Co-events that obey this are known as preclusive co-events.
\begin{definition}[Preclusive Co-event]
	A co-event $\phi \in \CS{t}$ is \textit{preclusive} with respect to the quantum measure $\mu_t$ if
	\begin{equation}
		\mu_t(E)=0 \implies \phi(E)=0 \,.
	\end{equation}
\end{definition}
Another condition that we will place on our co-events is that they agree with the past. That is, if the co-event at stage $t-1$ affirmed the event $E$, then we would expect co-events at stage $t$ to also affirm the same physical event, which is represented as $\Ext{E}$ at this stage. To be able to make this comparison between co-events at stage $t$ and $t-1$, we define a restriction map for co-events.
\begin{definition}[Restriction of a Co-event]
	The restriction map on $\CS{t}$, for $t>0$, is given by
	\begin{equation}
		\begin{alignedat}{2}
			\Prev:\CS{t} & \to \CS{t-1} 
			\\ \phi & \mapsto  \phi \Prev: \null & \EA{t-1} & \to \Z{2}
			\\ & &  E  & \mapsto \phi\Prev(E):=\phi\bigl( \Ext{E} \bigr) \, .
		\end{alignedat}
	\end{equation}
\end{definition}
The restriction removes any information about events that can only exists from stage $t$ onwards, but entirely preserves the strict past of $\phi$. Therefore, if $\phi \Prev = \phi'$, we say that $\phi$ is a \textit{prolongation} of $\phi'$ because it entirely agrees with $\phi'$ about all events that pertain to stage $t-1$ or earlier. In general, $\phi'$ will have several prolongations. For example, if $\phi'$ established that ``the hopper is at site $i$ or $i'$ at time $\tau=t-1$'' does happen, then any prolongation $\phi$ of $\phi'$ would also confirm that this event occurs, but would be free to choose whether the event ``the hopper is at site $i$ at both times $\tau=t-1$ and $\tau=t$'' happens or not.

Note that
\begin{listeqn}
	(\phi_1 \cdot \phi_2)\Prev = \phi_1\Prev \cdot \phi_2\Prev
	\AND (\phi_1 + \phi_2)\Prev = \phi_1\Prev + \phi_2\Prev \,.
\end{listeqn}
In addition,
\begin{align}
	\gamma^*\Prev(E)& =\gamma^*\bigl( \Ext{E} \bigr)
	\\ & = \begin{cases}
		1 & \quad \text{if }\gamma \in \Ext{E}
		\\ 0 & \quad \text{otherwise}
	\end{cases}
	\\ & = \begin{cases}
		1 & \quad \text{if }\gamma\Prev \in E
		\\ 0 & \quad \text{otherwise}
	\end{cases} \BY{ext_E_is_phys_equiv_to_E}
	\\ & = {\gamma\Prev}^*(E) \,,
\end{align}
which means in general
\begin{equation}
	E^*\Prev={E\Prev}^* \,. \label{monomial_restricted}
\end{equation}
In terms of the co-event in \eqref{example_coevent}, if $\gamma_1\Prev=\gamma_2\Prev=\gamma$ and $\gamma_3\Prev=\gamma'$ then
\begin{align}
	\phi\Prev & = \gamma^* + \gamma^* \cdot \gamma^* + {\gamma'}^*
	\\ & = \gamma^* + \gamma^* + {\gamma'}^*
	\\ & = 0 + {\gamma'}^*
	\\ & = {\gamma'}^* \,.
\end{align}

Finally, in our evolving co-event schemes we will become interested in restricting our co-events to be as simple or small as possible, in order to push us towards classical co-events, which are specified by a single history. ``Simple or small'' will come to mean that the co-event depends on few histories, in the same sense that the real function $f(x,y,z) := x+y$ only depends on the variables $x$ and $y$. Now to establish this, we want a notion of a partial difference on co-events. In particular, since $\EA{t}$ has a Boolean structure, a co-event can be equivalently considered to be a Boolean function. Boolean functions are used in Logic and Computing, and from this field \cite{1959_boolean_difference_Sheldon_Akers} we can borrow their choice for a partial difference.
\begin{definition}[Partial Difference]
	The partial difference of $\phi \in \CS{t}$ with respect to a history $\gamma \in \HS{t}$ is given by
	\begin{equation}
		\begin{aligned}
			\dif{\phi}{\gamma}: \EA{t} & \to \Z{2}
			\\  E & \mapsto \dif{\phi}{\gamma}(E) := \phi(E) + \phi(E + \gamma) \,.
		\end{aligned} \label{dif_definition}
	\end{equation}
	The partial difference with respect to an event $E \in \EA{t}$ is given by
	\begin{equation}
		\dif{}{E} := \prod \limits_{\gamma \in E} \dif{}{\gamma} \,,
	\end{equation}
	with
	\begin{equation}
		\dif{}{0} := 1 \,.
	\end{equation}
\end{definition}
Note that by applying \eqref{dif_definition} iteratively we find that for any $A,E \in \EA{t}$
\begin{equation}
	\dif{\phi}{A}(E) = \sum \limits_{B \in \Pow{A}}\phi(E+B) \,.\label{dif_E_identity}
\end{equation}
In addition, for any co-events $\phi_1,\phi_2 \in \CS{t}$ and $\gamma \in \HS{t}$
\begin{listeqn}
	\dif{}{\gamma}(\phi_1+\phi_2) = \dif{\phi_1}{\gamma} + \dif{\phi_2}{\gamma}
	\AND \dif{}{\gamma}(\phi_1 \cdot \phi_2) = \phi_1 \cdot \dif{\phi_2}{\gamma} + \dif{\phi_1}{\gamma} \cdot \phi_2 + \dif{\phi_1}{\gamma} \cdot \dif{\phi_2}{\gamma}\,. \label{dif_mult_rule}
\end{listeqn}
Moreover, the following relations that hold for all histories $\gamma,\gamma' \in \HS{t}$ and events $E,A \in \EA{t}$, show that the partial difference acts on a co-event as you would expect a differential operator to act on real polynomials:
\begin{gather}
	\dif{{\gamma'}^*}{\gamma} = \begin{cases}
		1 & \quad \text{if }\gamma'=\gamma
		\\ 0 & \quad \text{otherwise,} \label{dif_hist_with_gamma}
	\end{cases}
	\\ \dif{E^*}{\gamma} = \begin{cases}
		(E+\gamma)^* & \quad \text{if } \gamma \in E
		\\ 0 & \quad \text{otherwise,}
	\end{cases} \label{dif_E_with_gamma}
	\\ \dif{E^*}{A}  = \begin{cases}
		(E+A)^* & \quad \text{if } E \supseteq A
		\\ 0 & \quad \text{otherwise.}
	\end{cases} \label{dif_event_with_event}
\end{gather}
For example, for the co-event in \eqref{example_coevent}
\noMdent
\begin{align}
	& \phi = {\gamma_1}^* + {\gamma_1}^*{\gamma_2}^* + {\gamma_3}^*
	\\ \IMP & \dif{\phi}{\gamma_1} = 1 + {\gamma_2}^* + 0
	\\ \IMP & \dif{\phi}{\gamma_1+\gamma_2} = 0+1+0 \,,
\end{align}
\reMdent
whilst
\begin{equation}
	\dif{\phi}{\gamma_4} = 0 \,.
\end{equation}

Now, the full set of histories that the co-event depends on is known as its support.
\begin{definition}[Support]
	The support of a co-event $\phi \in \CS{t}$ is the event \footnote{The definition given here is equivalent to the usual QMT definition of the support of a co-event, where the support is the smallest event $S$ such that $\phi(E)=\phi(E \cdot S)$ for all events $E$ \cite{2007_co-events_Rafael_Sorkin}.}
	\begin{equation}
		\supp{\phi} := \setdef[\bigg]{\gamma \in \HS{t}}{ \dif{\phi}{\gamma} \neq 0} \,.
	\end{equation}
\end{definition}
For our example co-event in \eqref{example_coevent} the support would be $\set{ \gamma_1,\gamma_2,\gamma_3 }$. Note, that for any $\gamma \not\in \supp{\phi}$
\noMdent
\begin{align}
	0 &= \dif{\phi}{\gamma}(E)
	\\  &= \phi(E) + \phi(E+\gamma)
	\\ \IMP \phi(E) &= \phi(E+\gamma) \,. \label{independent_of_histories_outside_support}
\end{align}
\reMdent
We can use this to add or remove histories to the argument of $\phi$, if they are not in the support of $\phi$. In particular, we can use this to show that
\begin{equation}
	S \supseteq \supp{\phi} \implies \phi(E \cdot S) = \phi(E)  \quad \forall E \in \EA{t}\,, \label{co-event_supported_on_S}
\end{equation}
which will be a useful relation later on.

\subsection{Summary of Definitions}
Since we have introduced many definitions, we have included a summary of the concepts.
\begin{itemize}
	\item $\HS{t}$: History Space at stage $t$.
	\item $\gamma \in \HS{t}$: A history that describes the system up to stage $t$.
	\item $\gamma \Prev$: The history restricted to the previous stage.
	\item Extension: A history is an extension of a previous history if it restricts to that previous history.
	\item $\EA{t}$: The event algebra, which is the power set of $\HS{t}$.
	\item $E \in \EA{t}$: An event, which is a collection of histories.
	\item $A \cdot B$: The event ``$A$ and $B$''.
	\item $A + B$: The event ``$A$ or $B$ but not both''.
	\item $\gamma \in \EA{t}$: The event that only contains the single history $\gamma$.
	\item $\Ext{E}$: The extension of the event $E$ to the next stage, representing the same physical event but using later histories.
	\item $E\Prev$: The event $E$ restricted to the previous stage, loosing details that can only be represented at the current stage.
	\item $d_t(\gamma_1,\gamma_2)$: The decoherence matrix that captures the interference between histories.
	\item $D_t(A,B)$: The decoherence function taken by summing over $d_t$ for all pairs of histories in $A$ and $B$.
	\item $\mu_t(E)$: The quantum measure, the diagonal of $D_t$ and a weak generalisation of a classical measure.
	\item Null Event: An event with zero quantum measure.
	\item Null History: A history who's corresponding event has zero quantum measure, which can not interfere with other histories.
	\item $v_\gamma$: The vector which maps $\gamma$ to $1$ and everything else to $0$.
	\item $v_E$: The vector which maps histories in $E$ to $1$ and everything else to $0$.
	\item $\CS{t}$: The co-event space at stage $t$.
	\item $\phi \in \CS{t}$: A co-event representation of reality that exactly describes which events happen and do not happen.
	\item Preclusive: A co-event is preclusive if all null events do not happen.
	\item $\gamma^*$: A co-event where the history $\gamma$ happens and everything else follows classically.
	\item $E^*$: A monomial co-event that is the product of classical co-events for each history in $E$.
	\item $\phi\Prev$: The co-event $\phi$ restricted to the previous stage, loosing details that can only be represented at the current stage.
	\item Prolongation: A co-event is a prolongation of a previous co-event if it restricts to that previous co-event.
	\item $\dif{\phi}{\gamma}$: A difference operation on $\phi$, measuring how the co-event changes with the history $\gamma$.
	\item $\dif{\phi}{E}$: The above operation repeated for all histories in $E$.
	\item $\supp{\phi}$: The support of $\phi$, which contains all the histories $\phi$ depends on.
\end{itemize}

\section{Evolving Co-event Schemes} \label{evolving_schemes}
We have defined the basic objects of our stochastic theory $(\HS{t},\Prev,\EA{t},\mu_t)$, and the objects and concepts that derive from them. We now wish to define a \textit{co-event scheme} that determines which co-events in $\CS{t}$ are the allowed ones, given our theory. To begin with, we introduce a classical evolving co-event scheme.

\subsection{Classical Evolving Scheme}

In our theory, classical systems are characterised by a decoherence matrix that is completely diagonal, as such the quantum measure acts like a classical measure. We want to define a classical evolving scheme that works for such systems, but also manifests the intuitive features of classical physics. We want our classical co-event scheme to produce at each stage a set of allowed realities that do not conflict with classical intuition. Whilst there may be multiple allowed co-events at a given stage, only \textit{one} will be `chosen' to be the reality that is \textit{actually} expressed or experienced by the system. We explicitly do not provide any mechanism for this choice. Only once stage $t$'s single co-event has been chosen/enacted can the scheme continue for stage $t+1$. Thus, our scheme will iteratively produce a sequence of co-events $\CovS{t}$.

The first condition we might expect from a classical co-event scheme is that the co-events are classical. As we have already seen, classical co-events $\phi$ are characterised by a single history $\gamma \in \HS{t}$ by
\begin{equation}
	\phi=\gamma^* \,.
\end{equation}

In addition to being classical, we would not want zero probability events to be allowed to occur in our classical co-event scheme, so we would require that the co-events are preclusive. We denote the full set of preclusive co-events at stage $t$ by
\begin{equation}
	\Prec{t} := \setdef[\big]{\phi \in \CS{t}}{\forall E \in \EA{t} \ \mu_t(E)=0  \implies  \phi(E)=0} \,.
\end{equation}

So far we have only referred to the conditions the co-events must obey at every stage. However, if we choose the classical co-event characterised by ``the walker takes the path $0 \hopto 0 \hopto 1$'' at stage $2$, then we would not want to then choose the classical co-event characterised by ``the walker takes the path $1 \hopto 2 \hopto 2 \hopto 1$'' at the next stage. We want our new co-event to agree with the previous co-event about what happened in the past, in other words we want the new co-event $\CovS{t}$ to be a prolongation of the previous one $\CovS{t-1}$. We denote the set of all prolongations of the co-event $\CovS{t-1}$ by
\begin{equation}
	\Prol{t}\bigl(\CovS{t-1}\bigr):=\setdef[\big]{\phi \in \CS{t}}{\phi\Prev=\CovS{t-1}} \,.
\end{equation}
For convenience, we introduce the following shorthand, which is used only when no explicit argument is given:
\begin{gather}
	\Prol{t}  \equiv \Prol{t}\bigl(\CovS{t-1}\bigr) \,,
	\\ \PrecProl{t}  \equiv \Prec{t} \cap \Prol{t} \,.
\end{gather}

Now, we can combine these conditions to form the classical evolving co-event scheme. For our initial stage $t=0$ we have no past to appeal to, so we choose the initial co-event $\CovS{0}$ from the set
\begin{equation}
	\ClassicalS{0} :=\setdef{\gamma^*}{\gamma \in \HS{0}} \cap \Prec{0} \,.
\end{equation}
Then, at each stage $t>0$ we choose one co-event $\CovS{t}$ from the set
\begin{equation}
	\Classical{t} :=\setdef{\gamma^*}{\gamma \in \HS{t}} \cap \PrecProl{t} \,,
\end{equation}
which would depend on the previous element in our sequence: $\CovS{t-1}$. The ``next'' in the superscript highlights that this set is only formed once we have made a choice for the previous co-event, and this is the set of realities that can be chosen next.

To simplify the expression of this set, we label the histories that are selected to represent the co-events $\CovS{t}$ by $\HistS{t}$ such that
\begin{equation}
	\CovS{t}={\HistS{t}}^*\,.
\end{equation}
Now, the classical co-event ${\gamma}^*$ is a prolongation of the classical co-event ${\gamma'}^*$ iff $\gamma \Prev = \gamma'$. Therefore, our initial expression can reduce to
\begin{equation}
	\Classical{t}=\setdef[\big]{\gamma^*}{\gamma \in \HS{t},\ \gamma\Prev=\HistS{t-1}} \cap \Prec{t}\,.
\end{equation}

Moreover, we can use that $\mu_t$ is a classical measure. This means that an event $E$ will be null iff all of its subsets are also null. Therefore, a classical co-event $\gamma^*$ will be preclusive iff $\gamma$ is not a null history. So our set of co-events finally reduces to
\begin{equation}
	\Classical{t}=\setdef[\big]{\gamma^*}{\gamma \in \HS{t},\ \gamma\Prev=\HistS{t-1} \text{ and } \mu_t(\gamma) \neq 0} \,.
\end{equation}
Note, this set will not be empty because necessarily $\mu_{t-1}(\HistS{t-1}) \neq 0$, which means
\noMdent
\begin{align}
	& \mu_{t}\bigl( \Ext[\big]{\HistS{t-1}} \bigr) \neq 0
	\\ \IMP & \sum \limits_{\gamma \in \Ext{\HistS{t-1}}} \mu_{t}(\gamma) \neq 0 \,,
\end{align}
\reMdent
so there must exist an extension of $\HistS{t-1}$ that is not null. Therefore, this scheme always allows for a choice of co-event at the next stage.

This final result is indeed intuitive, for example for the random walker in a box, if at stage $2$ we have the classical co-event characterised by ``the walker takes the path $0 \hopto 1 \hopto 2$'' then the classical co-event scheme would allow us to choose $0 \hopto 1 \hopto 2 \hopto 2$ to be the characteristic history at the next stage but would forbid the choice of $0 \hopto 1 \hopto 0 \hopto 1$ on the grounds it is not an extension of the previous history, and would also forbid the choice of $0 \hopto 1 \hopto 2 \hopto 0$ on the grounds that this history has zero measure and is therefore impossible.

Note that we will not be associating a probability with each allowed co-event, instead, one can think of the set of allowed co-events at each stage as all the \textit{possible} realities that could be expressed at that stage. Whilst an obvious measure on the space of classical co-events would be given by the measure of the corresponding histories, non-classical co-events are not tied one-to-one with histories, or even events, so there is no such obvious choice for a measure. Moreover, using a purely possibilistic approach avoids some philosophical issues related to ontological probabilities, such as the question ``is our universe a likely one?''. In addition, physical probabilities for events could still emerge from a possibilistic co-event theory through repeated trials \cite{2009_emergence_of_probabilities_Yousef_Ghazi-Tabatabai_and_Petros_Wallden}, for example a theory that predicts there is no allowed reality that gives an infinite series of coin flips with more than half the outcomes being tails would predict $\mathrm{Prob}(\mathrm{tail}) \leq {1 \over 2}$. 

Now, for quantum systems we may wish to adopt the same scheme. However, there exist quantum measures that would give an empty set of classical scheme co-events. For example, in the usual description of the $n$-site hopper, by stage $t \sim n$ all histories lie within some null event, rendering all classical co-events non-preclusive \cite{2017_n-site_nirvana_Fay_Dowker_and_Vojtech_Havlicek_and_Cyprian_Lewandowski_and_Henry_Wilkes}. This covering of the history space with null events is what gives systems like the three-slit or GHZ experiments their anti-classical-realist flavour. So as we transfer to quantum systems we will have to loosen either the preclusive condition or the classical co-event condition. The former approach has been attempted in the past \cite{1994_restricted_preclusion_Rafael_Sorkin}, but this paper will follow the more common latter approach.

\subsection{Basic Evolving Scheme} \label{basic_evolving_scheme}

Now, to construct a general evolving scheme that also works for quantum systems we may wish to simply drop the classical co-event condition and choose our co-events from $\PrecProl{t}$. However, in general this set will contain a whole swath of very complex co-events, yet we still want our scheme to return to the classical co-events when it is presented with a classical measure. So we want some pressure that pushes us towards classical co-events. Since one characteristic of classical co-events is that their support only contains one history (it only depends on the characteristic history that defines it), one way to do this is to require that the support is as small as possible. Specifically, we introduce the concept of a \textit{minimal support}.
\begin{definition}[Minimal Support]
	A co-event $\phi \in \CS{t}$ is said to have a more minimal support than the co-event $\phi' \in \CS{t}$ iff
	\begin{equation}
		\supp{\phi} \subset \supp{\phi'} \,.
	\end{equation}
	Using this as a partial ordering for co-events in some set $\varphi \subseteq \CS{t}$, a co-event is said to be minimally supported in $\varphi$ if it is a minimal element in this ordering. Thus, we further define the minimal support operation to be 
	\begin{equation}
		\MinSupp{\varphi} := \setdef[\big]{\phi \in \varphi}{\not \exists \phi' \in \varphi \st \supp{\phi'} \subset \supp{\phi}} \,,
	\end{equation}
	which leaves only the minimally supported co-events in $\varphi$.
\end{definition}

We can now present the basic evolving scheme, as given by Rafael Sorkin. This scheme chooses a sequence of co-events, as in the classical scheme, except we choose the new co-event $\CovS{t}$ to be a \textit{minimally supported preclusive prolongation} of the previous co-event $\CovS{t-1}$. So for each stage $t>0$ we choose the new co-event from
\begin{equation}
	\Basic{t}:= \MinSupp[\big]{\PrecProl{t}} \,,
\end{equation}
which again depends on our choice of $\CovS{t-1}$. We also need to choose an initial co-event $\CovS{0}$. It would seem reasonable to choose $\CovS{0}$ from
\begin{equation}
	\BasicS{0} := \MinSupp[\big]{\Prec{0} \cap \setdef{\phi \in \CS{0}}{\phi(1)=1}} \, , \label{basic_0_def}
\end{equation}
where the second condition ensures that the co-events will all affirm the event $1 = \HS{0} =$``something happens'', and therefore also ensures that we do not choose the trivial co-event $0$. Now, this evolving scheme is referred to as `basic' because it only places a few constraints on the allowed sequences of co-events. We will briefly explore other schemes that add additional restrictions in Section \ref{more_evolving_schemes}. Note that in choosing the new co-events to be prolongations of the previous co-event, as opposed to, say, approximate prolongations, we are not allowing for any disputes about the past or any crystallizing effect \cite{2010_crystallizing_block_George_Ellis_and_Tony_Rothman}.

For the purpose of this paper, it is actually helpful to rephrase the scheme using more global objects. We iteratively generate
\begin{equation}
	\BasicS{t} := \bigcup \limits_{\phi \in \BasicS{t-1}} \MinSupp[\big]{\PrecProl{t}(\phi)} \,,
\end{equation}
which does not depend on some choice of $\CovS{t-1}$, and $\BasicS{0}$ is the same as in \eqref{basic_0_def}. Then an \textit{expressible} sequence of co-events $\CovS{t}$ is one that obeys
\begin{listeqn}
	\CovS{t}  \in \BasicS{t}
	\AND \CovS{t} \Prev = \CovS{t-1} \,.
\end{listeqn}
The set of expressible sequences of co-events is equivalent to the set of all co-event sequences that would be generated by following the previous single-sequence method. This alternative presentation will be useful for exposing properties possessed by all co-events generated by our basic evolving scheme.

We can similarly re-express the classical scheme. We iteratively generate
\begin{equation}
	\ClassicalS{t} := \bigcup \limits_{\phi \in \ClassicalS{t-1}} \Classical{t}(\phi) \,. \label{classical_set_def}
\end{equation}
And again, an \textit{expressible} sequence of co-events $\CovS{t}$ is required to obey
\begin{listeqn}
	\CovS{t}  \in \ClassicalS{t}
	\AND \CovS{t} \Prev = \CovS{t-1} \,.
\end{listeqn}
Note that the expression in \eqref{classical_set_def} reduces to
\begin{equation}
	\ClassicalS{t} = \setdef{\gamma^*}{\gamma \in \HS{t}, \ \mu_t(\gamma) \neq 0} \,, \label{classical_coevents_not_null_hist}
\end{equation}
which is the set of classical outcomes that do not affirm null events. So, as one would expect, for our classical scheme, if a history is not impossible, then there exists an allowed version of reality where it happens.

Later on, in Corollary \ref{return_to_classical}, we will show that when $\mu_t$ is a classical measure
\begin{equation}
	\BasicS{t} = \ClassicalS{t}\,,
\end{equation}
which means that our basic evolving scheme can fully reproduce the classical scheme.

\section{The Exclusion of Null Histories: Some Claims} \label{claims}

From our previous investigation, it is clear that the classical scheme only allows non-null histories to occur. Moreover, if we were to add or remove null histories from our history space the set of allowed co-events would be the same. This feature of the classical evolving scheme is useful for two reasons: we are free to extend our history space for convenience without physically changing the theory provided the new histories are null, as we did for the random walker in a box, and we can freely forget about histories that break the constraints of our system like histories that break the laws of motion. Moreover, a theory that does \textit{not} have this feature would find it difficult to protect itself from the pathology of a continuous demand to increase and generalise the histories in case the mere presence of a new, but non-interfering, history changes the results.

Therefore, we would certainly hope that our basic evolving scheme could also exclude null histories without physically effecting the theory. In general, the co-events will not be classical, so we can not reduce them to just one history, but we know that they only depend on the histories in their support. Therefore, if no co-event in $\BasicS{t}$ will ever contain a null history in their support then we can say that our scheme does not depend on null histories. Indeed, as we shall show later, this is true for our basic evolving scheme. Moreover, one can carefully show that if we add or remove null histories from a theory, with the condition that null histories that are added or removed must also have strictly null extensions\footnote{The author is unaware of any physical theory that would produce null histories will non-null extensions. Moreover, such histories can not exist classically, so would not be considered in the same class as classical null histories, such as those that break the laws of motion. And so there is no motivation to want our theory to be independent of null histories with non-null extensions.} that are also added or removed for consistency, then the new theory will produce an identical $\BasicS{t}$. To show this, one only needs to map co-events from one theory to the other, just by changing domain, which is possible because the co-event's support will exist in both domains. Then one can check that the mapped co-event still satisfies the minimally supported preclusive prolongation condition in the new theory.

However, to complete the above arguments we still need to prove Theorem \ref{zero_hist_not_in_min_prec_prol_supp}: that null histories do not enter the support of our basic evolving co-events. It is surprisingly complex to get this result, but to get there we will expose some features of the basic evolving scheme, which may be useful for further investigation.

\subsection{Expansion of a Co-event} \label{expansion_of_coevents}

First, we will prove the result that any co-event can be expanded around any event using the difference operation. The idea of expanding the co-event as a \textit{polynomial} of homomorphic co-events, which is equivalent to the below expansion when $X=0$, comes from within quantum measure theory \cite{2007_co-events_Rafael_Sorkin}. But the idea of writing the co-event as a Taylor expansion around any event comes from the theory of Boolean functions \cite{1973_boolean_differential_calculus_Andre_Thayse_and_Marc_Davio}. This more general expansion around $X$ is not used in the rest of this paper, but could prove useful for future works.
\begin{theorem} \label{expansion_around_X}
	Any co-event $\phi \in \CS{t}$ can be expanded around any event $X \in \EA{t}$ as
	\begin{equation}
		\phi(A) = \sum \limits_{E \in \EA{t}} \dif{\phi}{E}(X) E^*(A+X) \,.
	\end{equation}
\end{theorem}
\begin{proof}
	First we will prove this for the event $X=0$. First consider a co-event $\phi \in \CS{t}$ defined by
	\begin{equation}
		\phi := \sum \limits_{E \in \EA{t}} f(E) E^*
	\end{equation}
	for some function
	\begin{equation}
		f: \EA{t} \to \Z{2} \,.
	\end{equation}
	Suppose $\phi$ has another expansion in terms of the function $g: \EA{t} \to \Z{2}$:
	\noMdent
	\begin{alignat}{2}
		&& \phi &= \sum \limits_{E \in \EA{t}} g(E) E^*
		\\ &\IMP&  0 &= \sum \limits_{E \in \EA{t}}\bigl(f(E) + g(E)\bigr) E^* \,.
		\intertext{But each $E^*$ is a unique non-zero function for a given $E \in \EA{t}$. Therefore,}
		&& 0 & = f(E) + g(E) \quad \forall E \in \EA{t}
		\\ &\IMP \makebox[0pt][l]{$\displaystyle g=f\,.$} 
	\end{alignat}
	\reMdent
	Therefore, this expansion is unique. Since $f$ is essentially a co-event there are $\vert \CS{t} \vert$ such expansions, all of which are unique co-events, which exhausts all the co-events in $\CS{t}$.
	
	Now, to prove the form of the function $f$, consider
	\begin{align}
		\dif{\phi}{A}(0) &= \sum \limits_{E \in \EA{t}} f(E) \dif{E^*}{A}(0)
		\\ &= \sum \limits_{E \supseteq A} f(E)(E+A)^*(0) \BY{dif_event_with_event}.
		\intertext{But $B^*(0)=0$ if $B \neq 0$, otherwise $0^*(0)=1$. Therefore, the only potential non-zero term is when $E=A$, which means}
		\dif{\phi}{A}(0) &= f(A) \,.
	\end{align}
	Thus, the coefficients of the expansion are given by the partial difference\footnote{Inspired by the form of $f$, we can write down the operation $D:\CS{t} \to \CS{t}$, characterised by $D(\phi)(E)=\dif{\phi}{E}(0)$. $D$ is a dual operation, in other words $D(D(\phi))=\phi$. To see this consider the co-event expansion where $\phi(E)$ is the co-efficient of $E^*$, instead of $D(\phi)(E)$, and it will be equal to $D(\phi)$.}.

	Now, for a general $X$ consider $\phi'$ given by
	\begin{equation}
		\phi'(A)=\phi(A+X)\,. \label{phi_prime_def}
	\end{equation}
	Note that
	\begin{align}
		\dif{\phi'}{E}(0) &= \sum \limits_{A \in \Pow{E}} \phi'(A) \BY{dif_E_identity}
		\\ &= \sum \limits_{A \in \Pow{E}} \phi(A+X)
		\\ &= \dif{\phi}{E}(X) \BY{dif_E_identity} . \label{dif_phi_prime}
	\end{align}
	Therefore,
	\noMdent
	\begin{alignat}{2}
		&&\phi'(A) &= \sum \limits_{E \in \EA{t}} \dif{\phi'}{E}(0) E^*(A)
		\\ && &=\sum \limits_{E \in \EA{t}} \dif{\phi}{E}(X) E^*(A) \BY{dif_phi_prime}
		\\ &\IMP& \phi(A)& =\sum \limits_{E \in \EA{t}} \dif{\phi}{E}(X) E^*(A+X) \BY{phi_prime_def}.
	\end{alignat}
	\reMdent
\end{proof}

As one might expect, the support of a co-event, consists of exactly the histories that form its polynomial expansion.

\begin{lemma} \label{supp_is_events_in_polynomial}
	Consider any $\phi \in \CS{t}$ in its polynomial form
	\begin{equation}
		\phi = \sum_{i \in I} {E_i}^* \,, \label{simple_plynomial_expansion}
	\end{equation}
	where $I$ is a set of indices and each event $E_i$ is unique. Then
	\begin{equation}
		\supp{\phi}=\bigcup \limits_{i \in I} E_i \,.
	\end{equation}
\end{lemma}
\begin{proof}
	Note that we can always write $\phi$ in the form of \eqref{simple_plynomial_expansion} by expanding around $0$ using Theorem \ref{expansion_around_X}, and only keeping the $E^*$ terms for which $\dif{\phi}{E}(0)=1$.

	Now, consider any $\gamma \in \HS{t}$, we find
	\begin{align}
		\dif{\phi}{\gamma} &= \sum_{i \in I} \dif{{E_i}^*}{\gamma}
		\\ &= \sum_{i \in I} \begin{cases}
			(E_i+\gamma)^* & \quad \text{if } \gamma \in E_i
			\\ 0 & \quad \text{otherwise}
		\end{cases} \BY{dif_E_with_gamma}.
	\end{align}
	Note that each $(E_i+\gamma)^*$ is unique for different $E_i$ and so these terms can not cancel. Therefore,
	\begin{equation}
		\dif{\phi}{\gamma} \neq 0 \iff \gamma \in \bigcup \limits_{i \in I} E_i \,.
	\end{equation}
\end{proof}
For example, consider the co-event from \eqref{example_coevent}:
\begin{equation}
	\phi={\gamma_1}^* + {\gamma_1}^* \cdot {\gamma_2}^* + {\gamma_3}^* \,.
\end{equation}
If we take the partial difference with respect to $\gamma_1$, $\gamma_2$ or $\gamma_3$ we get a non-zero answer, but if we take the partial difference with respect to some other history then we get $0$. So $\supp{\phi}=\set{\gamma_1,\gamma_2,\gamma_3}$, which consists of all the histories in its expansion.

If we are interested in a particular history $\gamma$, a useful trick, which we will use later on, is to split the co-event into terms that contain a $\gamma^*$ term and those that do not.

\begin{lemma} \label{split_coevent_wrt_hist}
	For any given $\gamma \in \HS{t}$, any co-event $\phi \in \CS{t}$ can be written as
	\begin{equation}
		\phi=\gamma^* \cdot \phi_1 + \phi_2 \,,
	\end{equation}
	where
	\begin{listeqn}
		\phi_1 :=\dif{\phi}{\gamma} \AND
		\phi_2 := \phi + \gamma^* \cdot \phi_1 \,.
	\end{listeqn}
	Moreover,
	\begin{listeqn}
		\supp{\phi_1} \subseteq (\supp{\phi} \setminus \gamma) \AND
		\supp{\phi_2} \subseteq (\supp{\phi} \setminus \gamma) \,.
	\end{listeqn}
\end{lemma}
\begin{proof}
	As in Lemma \ref{supp_is_events_in_polynomial}, consider $\phi$'s polynomial expansion in the form
	\begin{align}
		\phi &= \sum \limits_{i \in I} {E_i}^*
		\\ &= \sum \limits_{\substack{i \in I \st \\ \gamma \in E_i}} {E_i}^* + \sum \limits_{\substack{i \in I \st\\ \gamma \not \in E_i}} {E_i}^*
		\\ &= \gamma^* \sum \limits_{\substack{i \in I \st \\ \gamma \in E_i}} (E_i+\gamma)^* + \sum \limits_{\substack{i \in I \st\\ \gamma \not \in E_i}} {E_i}^* \label{long_expansion_of_phi}
		\\ &= \gamma^* \sum \limits_{i \in I} \dif{{E_i}^*}{\gamma} + \sum \limits_{\substack{i \in I \st\\ \gamma \not \in E_i}} {E_i}^* \BY{dif_E_with_gamma}
		\\ &= \gamma^* \dif{\phi}{\gamma} +  \sum \limits_{\substack{i \in I \st\\ \gamma \not \in E_i}} {E_i}^*
		\\ &= \gamma^* \phi_1 + \phi_2 \,.
	\end{align}
	Finally, note that the sums in \eqref{long_expansion_of_phi} are the polynomial expansions of $\phi_1$ and $\phi_2$ respectively, and neither contain any $\gamma^*$ terms. Therefore, by Lemma \ref{supp_is_events_in_polynomial}, neither $\phi_1$ nor $\phi_2$ contain $\gamma$ in their supports. Moreover, these sums are only made of events or subsets of events used in $\phi$'s polynomial expansion. Therefore, again by Lemma \ref{supp_is_events_in_polynomial}, both $\phi_1$ and $\phi_2$'s supports must be subsets of $\phi$'s support.
\end{proof}

For example, we can write the co-event from \eqref{example_coevent} as
\begin{alignat}{3}
	\phi& = {\gamma_1}^* & & \cdot  ( 1 + {\gamma_2}^*) & & +  {\gamma_3}^*
	\\ & =  {\gamma_2}^* && \cdot  {\gamma_1}^* && +  ({\gamma_1}^* + {\gamma_3}^*)
	\\ & = {\gamma_3}^* && \cdot  1 && +  ({\gamma_1}^* + {\gamma_1}^* \cdot {\gamma_2}^*)
	\\ & = {\gamma_4}^*  && \cdot  0 && + ({\gamma_1}^*  + {\gamma_1}^* \cdot {\gamma_2}^* + {\gamma_3}^*) \,.
\end{alignat}

Finally, we can show that the polynomial form of a prolongation is constrained by the co-event it restricts to.

\begin{theorem} \label{prolongation_expansion_condition}
	Consider any $\phi \in \CS{t-1}$ in its polynomial form
	\begin{equation}
		\phi = \sum_{i \in I} {E_i}^* \,,
	\end{equation}
	where $I$ is a set of indices and each event $E_i$ is unique. A co-event $\phi' \in \CS{t}$ is a prolongation of $\phi$ iff its polynomial expansion can be written (non-uniquely) in the form
	\begin{equation}
		\phi' = \sum_{i \in I} {A_i}^* + \sum_{j \in J} ({B_j}^* + {C_j}^*) \,,
	\end{equation}
	where $J$ is some set of indices, each event $A_i$, $B_j$ and $C_j$ is unique,
	\begin{listeqn}
		{A_i}\Prev =E_i
		\AND {B_j}\Prev = {C_j}\Prev \,.
	\end{listeqn}
\end{theorem}
\begin{proof}
	Consider any $\phi' \in \CS{t}$ with its polynomial expansion arranged as
	\begin{equation}
		\phi' = \sum_{i \in I'} {A_i}^* + \sum_{j \in J} ({B_j}^* + {C_j}^*) + \sum_{k \in K} {D_k}^* \,,
	\end{equation}
	where $I' \subseteq I$, $J$ and $K$ are some set of indices, all the events are unique,
	\begin{listeqn}
		{A_i}\Prev =E_i \label{A_restricted_is_E}
		\COM {B_j}\Prev = {C_j}\Prev \label{B_and_C_restricted_agree}
		\COM {D_k}\Prev \neq E_i \quad \forall i \in I\setminus I'
		\AND {D_k}\Prev \neq {D_{k'}}\Prev \quad \forall k' \neq k \,.
	\end{listeqn}
	Note that this is indeed general. To construct this we first find events in the polynomial expansion of $\phi'$ that restrict to $E_i$ for each $i \in I$, if there exists at least one we choose one (it does not matter which one) to be $A_i$, if we do not find one then necessarily $I' \subset I$. With the rest of the events in the polynomial expansion we pair them up (non-uniquely) such that their restrictions agree and label them $B_j$ and $C_j$, if there are none then $J=\emptyset$. Finally, the left over events in the polynomial expansion form $D_k$ and $K$, which necessarily can not be paired up with each other and do not restrict to any $E_i$ for $i$ outside of $I'$. Then this co-event's restriction is
	\begin{align}
		\phi'\Prev &= \sum_{i \in I'} {{A_i}^*}\Prev + \sum_{j \in J} \bigl({{B_j}^*}\Prev + {{C_j}^*}\Prev\bigr) + \sum_{k \in K} {{D_k}^*}\Prev
		\\ &=\sum_{i \in I'} {A_i\Prev}^* + \sum_{j \in J} \bigl({B_j\Prev}^* + {C_j\Prev}^*\bigr) + \sum_{k \in K} {D_k\Prev}^* \BY{monomial_restricted}
		\\ &=\sum_{i \in I'} {E_i}^* + \sum_{j \in J} 0 + \sum_{k \in K} {D_k\Prev}^* \BY{A_restricted_is_E} \text{ and \eqref{B_and_C_restricted_agree}}
		\\ &=\phi + \sum_{i \in I \setminus I'} {E_i}^* + \sum_{k \in K} {D_k\Prev}^* \,.
	\end{align}
	Now, note that the ${D_k\Prev}^*$ terms can not cancel with each other or the ${E_i}^*$ terms, and the ${E_i}^*$ can not cancel with each other because they are unique. Therefore, $\phi'$ is a prolongation of $\phi$ iff $I'=I$ and $K=\emptyset$.
\end{proof}

For example, we now know from Theorem \ref{prolongation_expansion_condition} that the co-event
\begin{equation}
	{\gamma_1'}^*\cdot {\gamma_1''}^* + {\gamma_1'}^* \cdot {\gamma_2'}^* \cdot {\gamma_2''}+ {\gamma_3'}^* + \bigl({\gamma_1''}^*\cdot {\gamma_2'}^*\cdot {\gamma_2''}^* + {\gamma_1'}^* \cdot {\gamma_2'}^*\bigr) + \bigl({\gamma_1'}^*\cdot {\gamma_4'}^* + {\gamma_1''}^*\cdot {\gamma_4'}^*\bigr)
\end{equation}
will be a valid prolongation of the co-event from \eqref{example_coevent} if
\begin{listeqn}
	\gamma_1' \Prev =\null  \gamma_1''\Prev  = \gamma_1
	\COM \gamma_2' \Prev =\null  \gamma_2''\Prev = \gamma_2
	\COM \gamma_3'\Prev = \gamma_3
	\AND \gamma_4'\Prev =\gamma_4 
\end{listeqn}
because the first three terms will restrict to produce the terms in the previous co-event and the terms inside the brackets will cancel when they are restricted. Note that each of the terms in the first bracket could be swapped with the second term and play the same role because they all restrict to the same monomial. Moreover, we can see that the prolongation will be at least as complex as the previous co-event, a side effect of this is that the support of a co-event can not reduce in size via prolongation. Therefore, a non-classical co-event can never produce a classical co-event through prolongation.

\subsection{Supports of the Evolving Scheme}

We can now begin to explore the properties of supports of co-events that are produced by the basic evolving scheme. We start with a claim given by Rafael Sorkin\footnote{Given during a discussion at Perimeter Institute in July 2016.}.
\begin{lemma} \label{aff_den_pair_intersect_supp}
	Given a co-event $\phi \in \CS{t}$ with $\supp{\phi}=S$,
	\begin{equation}
		\phi(A) \neq \phi(D) \implies (A+D)\cdot S \neq 0 \, .
	\end{equation}
\end{lemma}
\begin{proof}
	\noMdent
	\begin{alignat}{2}
		&& (A+D)\cdot S &=0
		\\ &\IMP& A \cdot S &= D \cdot S
		\\ &\IMP& \phi(A \cdot S) &= \phi(D \cdot S)
		\\ &\IMP& \phi(A) &= \phi(D) \BY{co-event_supported_on_S} .
	\end{alignat}
\end{proof}

Now, the preclusive condition in the basic evolving co-event scheme forces every new co-event to deny the null events. In addition, the prolongation condition forces all the new co-events to agree with the previous co-event about previous events. So our scheme provides us with a set of events that must be denied, and another set of events that must be affirmed. Using this and Lemma \ref{aff_den_pair_intersect_supp} we can determine the supports of these new co-events. To do so, we define the following set of events:
\begin{gather}
	\FA{t} := \setdef[\big]{A \in \EA{t}}{\exists A' \in \EA{t-1} ,\ \CovS{t-1}(A')=1 \text{ and } A=\Ext{A'}} \,,
	\\ \FDprol{t} := \setdef[\big]{D \in \EA{t}}{\exists D' \in \EA{t-1} ,\ \CovS{t-1}(D')=0 \text{ and } D=\Ext{D'}}\,,
	\\ \FDprec{t} := \setdef{ P \in \EA{t}}{\mu_t(P)=0 } \,,
	\\ \FD{t} := \FDprec{t} \cup \FDprol{t} \,.
\end{gather}
So $\FA{t}$ is the set of all events that must be \textbf{A}ffirmed by the new co-event (by prolongation), $\FDprec{t}$ is the set of all null events that must be \textbf{D}enied by the new co-event through \textbf{preclusion}, and $\FDprol{t}$ is the set of all events that must be \textbf{D}enied by \textbf{prolongation}. Note that since extensions of null events are also null by \eqref{qmeasure_consistency}, if the previous co-event $\CovS{t-1}$ is preclusive then $\FDprec{t}$ and $\FDprol{t}$ will overlap, and importantly $\FA{t}$ and $\FD{t}$ do not overlap. These constructions lead us to our next claim.
\begin{lemma} \label{prec_prol_iff_supp_int}
	If $\CovS{t-1} \in \Prec{t-1}$, then for any $S \in \EA{t}$
	\begin{equation}
		\exists \, \phi \in \PrecProl{t} \text{ with } \supp{\phi} \subseteq S \iff  (A+D) \cdot S \neq 0 \quad \forall A \in \FA{t}, \ D \in \FD{t} \, .
	\end{equation}
\end{lemma}
\begin{proof}
	First, consider a $\phi \in \PrecProl{t}$ with $\supp{\phi} \subseteq S$. Since $\phi$ is a prolongation it will affirm all the events in $\FA{t}$ and will deny all the events in $\FDprol{t}$. Moreover, since $\phi$ is also preclusive it will also deny all the events in $\FDprec{t}$. Therefore, $\phi$ denies all the events in $\FD{t}$ and affirms all the events in $\FA{t}$ (note that this is consistent, i.e. there is no overlap between $\FA{t}$ and $\FD{t}$ because the previous co-event was preclusive). Therefore,
	\noMdent
	\begin{align}
		& \phi(A) \neq \phi(D)\quad \forall A \in \FA{t}, \  D \in \FD{t}
		\\ \IMP &(A+D) \cdot \supp{\phi} \neq 0 \quad \forall A \in \FA{t}, \  D \in \FD{t} \BYlem{aff_den_pair_intersect_supp}
		\\ \IMP & (A+D) \cdot S \neq 0 \quad \forall A \in \FA{t}, \ D \in \FD{t} \,.
	\end{align}
	Now we will show the reverse. Consider an event $S$ such that
	\begin{alignat}{2}
		&& (A+D) \cdot S & \neq 0 \quad \forall A \in \FA{t}, \  D \in \FD{t}
		\\ &\IMP& A \cdot S & \neq D \cdot S \quad \forall A \in \FA{t}, \  D \in \FD{t} \, . \label{A_S_noteq_D_S}
	\end{alignat}
	\reMdent
	We can construct a particular co-event $\phi \in \CS{t}$ as
	\begin{align}
		\phi(E) = \begin{cases}
			\phi(E \cdot S) & \quad \text{if } E \not \subseteq S
			\\ 1 & \quad \text{if } \exists A \in \FA{t} \st A \cdot S = E
			\\ 0 & \quad \text{if } \exists D \in \FD{t} \st D \cdot S = E
			\\ 0 & \quad \text{otherwise.}
		\end{cases}
	\end{align}
	This definition is consistent since the first condition does not overlap with the last three, and the second and third conditions do not overlap by \eqref{A_S_noteq_D_S}.
	
	Now, consider any $\gamma \not \in S$. The construction of $\phi$ ensures that $\phi(E)=\phi(E \cdot S)$ for all events $E \in \EA{t}$. Therefore,
	\begin{align}
		\dif{\phi}{\gamma}(E) &= \phi(E) + \phi(E+\gamma)
		\\ &= \phi(E \cdot S) + \phi\bigl((E+\gamma) \cdot S\bigr)
		\\ &= \phi(E \cdot S) + \phi(E \cdot S)
		\\ &=0 \,.
	\end{align}
	Therefore, all the histories outside of $S$ are not in the support of $\phi$, which means $\supp{\phi} \subseteq S$. Moreover, $\phi$ affirms all events in $\FA{t}$ and denies all events in $\FD{t}$, and so is in $\PrecProl{t}$ as required. Note that we could have split and altered the ``otherwise'' condition without effecting the proof.
\end{proof}

We have shown that the support of a preclusive prolongation must overlap with all $(A+D)$ terms. We can now go one step further and characterise when this preclusive prolongation is also minimally supported.

\begin{lemma} \label{histories_in_min_supp}
	If $\CovS{t-1}\in \Prec{t-1}$ and $\phi \in \PrecProl{t}$ with $S=\supp{\phi}$, then
	\begin{equation}
		\phi \in \MinSupp[\big]{\PrecProl{t}} \iff \forall \gamma \in S \  \exists A \in \FA{t},\ D \in \FD{t} \st (A+D)\cdot S = \gamma \,.
	\end{equation}
\end{lemma}
\begin{proof}
	Note that if $S=0$ the claim is true since $\phi$ must be minimally supported and there are no histories in the support needed to satisfy the condition on the right-hand side. Now, considering $S\neq 0$, suppose
	\noMdent
	\begin{align}
		& \phi \not \in \MinSupp[\big]{\PrecProl{t}}
		\\ \IFF & \exists \phi' \in \PrecProl{t} \text{ with } \supp{\phi'} \subset S
		\\ \IFF & \exists \gamma \in S \st \exists \phi' \in \PrecProl{t} \text{ with } \supp{\phi'} \subseteq S+\gamma
		\\ \IFF & \exists \gamma \in S \st (A+D) \cdot (S+\gamma) \neq 0 \quad \forall A \in \FA{t},\ D \in \FD{t} \BYlem{prec_prol_iff_supp_int}.
	\end{align}
	Negating what we have shown so far, we find
	\begin{align}	
		& \phi \in \MinSupp[\big]{\PrecProl{t}}
		\\ \IFF & \forall \gamma \in S \  \exists A \in \FA{t},\ D \in \FD{t} \st (A+D) \cdot (S+\gamma) = 0
		\\ \IFF & \forall \gamma \in S \  \exists A \in \FA{t},\ D \in \FD{t} \st (A+D) \cdot S = (A+D) \cdot \gamma \,.
	\end{align}
	Now, the right-hand side of the equality is either the event $\gamma$ or $0$ because the event $\gamma$ only contains one history. However, suppose
	\begin{align}
		& \exists A \in \FA{t},\ D \in \FD{t} \st (A+D) \cdot S = 0
		\\ \IMP & \not \exists \phi' \in \PrecProl{t} \text{ with } \supp{\phi'} \subseteq S \BYlem{prec_prol_iff_supp_int}.
	\end{align}
	\reMdent
	This contradicts with $\phi$ itself being a preclusive prolongation. Therefore,
	\begin{equation}
		\phi \in \MinSupp[\big]{\PrecProl{t}} \iff  \forall \gamma \in S \  \exists A \in \FA{t},\ D \in \FD{t} \st (A+D) \cdot S = \gamma \,.
	\end{equation}
\end{proof}

Now, consider any co-event $\phi \in \CS{t}$. By the definition of the support, for each $\gamma \in \supp{\phi}$ there exists $E \in \EA{t}$ such that
\begin{align}
	1&= \dif{\phi}{\gamma}(E)
	\\ &= \phi(E) +\phi(E+\gamma)\,.
\end{align}
This means that for histories $\gamma$ in the support of $\phi$ there must exists an event $E$ such that either $E$ or $(E+\gamma)$ is affirmed, whilst the other event is denied. When $\phi$ is a preclusive prolongation, we can choose this event to be an event from $\FD{t}$. But as we shall see in the following theorem, if $\phi$ is also minimally supported then we can go further and choose this event to also be null. This relationship highlights the necessary role that null events play in the basic evolving scheme in bringing more histories into the support. 

\begin{theorem} \label{dif_prec_one_iterate}
	If $\CovS{t-1} \in \Prec{t-1}$, $\phi \in \MinSupp[\big]{\PrecProl{t}}$	and
	\begin{equation}
		\forall \gamma' \in \supp[\big]{\CovS{t-1}} \ \exists P \in \FDprec{t-1} \st \dif{\CovS{t-1}}{\gamma'}(P)=1\,,
	\end{equation}
	then
	\begin{equation}
		\forall \gamma \in \supp{\phi} \ \exists Q \in \FDprec{t} \st \dif{\phi}{\gamma}(Q)=1 \,.
	\end{equation}
\end{theorem}
\begin{proof}
	Let $S=\supp{\phi}$. To show that this is true for all $\gamma \in S$, we will split our proof into three cases:
	\begin{itemize}
		\item[1.]\hypertarget{case_1}{} $\Ext{\gamma\Prev}\cdot S \supset \gamma \,.$
		\item[2.]\hypertarget{case_2}{} $\Ext{\gamma\Prev}\cdot S = \gamma$ and $\gamma\Prev \in \supp[\big]{\CovS{t-1}} \,.$
		\item[3.]\hypertarget{case_3}{} $\Ext{\gamma\Prev}\cdot S = \gamma$ and $\gamma\Prev \not\in \supp[\big]{\CovS{t-1}} \,.$
	\end{itemize}
	First consider case \hyperlink{case_1}{1}. From Lemma \ref{histories_in_min_supp}, we know that $\exists A \in \FA{t}$ and $Q \in \FD{t}$ such that
	\begin{equation}
		(A+Q) \cdot S = \gamma\,. \label{A_and_Q_gives_hist}
	\end{equation}
	From the definition of $\FA{t}$, we know that there exists $A' \in \EA{t-1}$ such that $A=\Ext{A'}$. Now, suppose
	\begin{equation}
		Q \in \FDprol{t}\,, \label{suppose_Q_in_prol}
	\end{equation}
	then there also exists $Q' \in \EA{t-1}$ such that $Q=\Ext{Q'}$. Then
	\noMdent
	\begin{alignat}{2}
		&& \Ext{A'+Q'} \cdot S & = \gamma
		\\ &\IMP& \Ext{\gamma\Prev} \cdot \Ext{A'+Q'} \cdot S & = \Ext{\gamma\Prev} \cdot \gamma
		\\ &\IMP& \Ext[\big]{\gamma\Prev \cdot (A'+Q')} \cdot S & = \gamma \,.
	\end{alignat}
	\reMdent
	Now, the term in $\Ext{}$ is either $\gamma\Prev$ or $0$ because the event $\gamma\Prev$ only contains one history. However, if it was $0$ then the left-hand side could not agree with the right-hand side. Therefore, the term must be $\gamma\Prev$, thus
	\begin{equation}
		\Ext{\gamma\Prev} \cdot S  = \gamma \,.
	\end{equation}
	But this is a contradiction with the condition for case \hyperlink{case_1}{1}. Therefore, \eqref{suppose_Q_in_prol} is false, meaning $Q \in \FDprec{t}$. In addition,
	\begin{align}
		\dif{\phi}{\gamma}(Q) &= \phi(Q) + \phi(Q+\gamma)
		\\ &= \phi(Q+\gamma) \qquad \text{since }\phi \in \Prec{t}
		\\ &= \phi\bigl((Q+\gamma) \cdot S\bigr) \BY{co-event_supported_on_S}
		\\ &= \phi(A \cdot S) \BY{A_and_Q_gives_hist}
		\\ &= \phi(A) \BY{co-event_supported_on_S}
		\\ &= 1 \qquad \text{since } \phi \in \Prol{t} \,.
	\end{align}
	Now, consider case \hyperlink{case_2}{2}. Since $\gamma\Prev \in \supp{\CovS{t-1}}$, we know $\exists P \in \FDprec{t-1}$ such that
	\begin{align}
		1 &= \dif{\CovS{t-1}}{\gamma\Prev}(P)
		\\ &= \CovS{t-1}(P) + \CovS{t-1}(P+\gamma\Prev)
		\\ &= \CovS{t-1}(P+\gamma\Prev) \qquad \text{since }\CovS{t-1} \in \Prec{t-1}\,. \label{prevcov_affirmed_prec_plus_hist}
	\end{align}
	Let $Q=\Ext{P}$. By \eqref{qmeasure_consistency}, $Q$ is also null and therefore in $\FDprec{t}$. Now,
	\begin{align}
		\dif{\phi}{\gamma}(Q) & =\phi(Q) + \phi(Q + \gamma)
		\\ & = \phi(Q + \gamma) \qquad \text{since }\phi \in \Prec{t}
		\\ & = \phi(Q \cdot S + \gamma \cdot S) \BY{co-event_supported_on_S}
		\\ & = \phi(Q \cdot S + \Ext{\gamma\Prev} \cdot S) \qquad \text{by the first condition on case \hyperlink{case_2}{2}}
		\\ & = \phi(Q+ \Ext{\gamma\Prev}) \BY{co-event_supported_on_S}
		\\ & = \phi\bigl( \Ext{P+\gamma\Prev} \bigr)
		\\ & = \phi\Prev(P+\gamma\Prev)
		\\ & = \CovS{t-1}(P+\gamma\Prev) \qquad \text{since }\phi \in \Prol{t}
		\\ & = 1 \BY{prevcov_affirmed_prec_plus_hist}\,.
	\end{align}
	Finally, consider case \hyperlink{case_3}{3}. We will use Lemma \ref{split_coevent_wrt_hist} to split $\phi$ into
	\begin{equation}
		\phi=\gamma^*\cdot \phi_1 + \phi_2\,,
	\end{equation}
	where $\phi_1=\dif{\phi}{\gamma}$. From the same lemma we know that
	\begin{equation}
		\supp{\phi_2} \subseteq (S+\gamma) \,. \label{phi2_supp_is_smaller}
	\end{equation}
	Now, consider any $E \in \EA{t-1}$, then
	\noMdent
	\begin{alignat}{2}
		&\Mdentspace& \phi_2\Prev(E+\gamma\Prev) &= \phi_2\bigl( \Ext{E+\gamma\Prev} \bigr)
		\\ && &= \phi_2\bigl( \Ext{E} \cdot (S+\gamma) + \Ext{\gamma\Prev} \cdot (S+\gamma) \bigr) \BY{phi2_supp_is_smaller}\text{ and \eqref{co-event_supported_on_S}}
		\\ && &= \phi_2\bigl( \Ext{E} \cdot (S+\gamma) + \gamma+\gamma \bigr) \qquad \text{by the first condition on case \hyperlink{case_3}{3}}
		\\ && &= \phi_2\bigl( \Ext{E} \bigr) \BY{co-event_supported_on_S}
		\\ && &= \phi_2\Prev(E)
		\\ & \IMP \makebox[0pt][l]{$\displaystyle \dif{\phi_2\Prev}{\gamma\Prev} = 0 \,.$}\label{no_prev_hist_in_phi2_prev}
	\end{alignat}
	\reMdent
	Similarly, we can show by the same method that
	\begin{equation}
		\dif{\phi_1\Prev}{\gamma\Prev} = 0 \,.\label{no_prev_hist_in_phi1_prev}
	\end{equation}
	Now, by the second condition in case \hyperlink{case_3}{3}
	\begin{align}
		0 &= \dif{\CovS{t-1}}{\gamma\Prev}
		\\ &= \dif{\phi\Prev}{\gamma\Prev} \qquad \text{since }\phi \in \Prol{t}
		\\ &= \dif{}{\gamma\Prev}({\gamma\Prev}^*\cdot \phi_1\Prev+\phi_2\Prev) \BY{monomial_restricted}
		\\ &= \dif{{\gamma\Prev}^*}{\gamma\Prev} \cdot \phi_1\Prev + {\gamma\Prev}^* \cdot \dif{\phi_1\Prev}{\gamma\Prev} + \dif{{\gamma\Prev}^*}{\gamma\Prev} \cdot \dif{\phi_1\Prev}{\gamma\Prev} + \dif{\phi_2\Prev}{\gamma\Prev} \BY{dif_mult_rule}
		\\ &= 1 \cdot \phi_1\Prev + 0 + 0 + 0\BY{dif_hist_with_gamma}\text{, \eqref{no_prev_hist_in_phi1_prev} and \eqref{no_prev_hist_in_phi2_prev}} . \label{phi1_prev_is_zero}
	\end{align}
	Therefore,
	\begin{align}
		\phi_2\Prev &= \phi\Prev + {\gamma\Prev}^* \cdot \phi_1\Prev
		\\  &= \phi\Prev \BY{phi1_prev_is_zero}
		\\  &= \CovS{t-1} \qquad \text{since }\phi \in \Prol{t}\,.
	\end{align}
	Therefore, $\phi_2 \in \Prol{t}$. Now, suppose
	\begin{equation}
		\phi_1(P) := \dif{\phi}{\gamma}(P) = 0\quad \forall P \in \FDprec{t} \,. \label{phi2_is_prec}
	\end{equation}
	Then, for any $P \in \FDprec{t}$
	\begin{align}
		\phi_2(P) &= \phi(P) + \gamma^*(P) \cdot \phi_1(P)
		\\ &= 0 \qquad \text{since } \phi \in \Prec{t} \text{ and by \eqref{phi2_is_prec}.}
	\end{align}
	 Therefore, $\phi_2 \in \PrecProl{t}$, but $\supp{\phi_2} \subset S$ which is a contradiction with $\phi$ being a minimally supported preclusive prolongation. Therefore, \eqref{phi2_is_prec} is false, meaning $\exists Q \in \FDprec{t}$ such that
	\begin{equation}
		\dif{\phi}{\gamma}(Q)=1\,.
	\end{equation}
\end{proof}

\subsection{Null Events}

We can now prove the main claim we have been working towards.
\begin{theorem} \label{zero_hist_not_in_min_prec_prol_supp}
	\begin{equation}
		\mu_t(\gamma)=0 \implies \gamma \not\in \supp{\phi} \quad \forall \phi \in \BasicS{t} \,.
	\end{equation}
\end{theorem}
\begin{proof}
	First consider
	\begin{align}
		\phi & \in \BasicS{0}
		\\ &= \MinSupp[\big]{\Prec{0} \cap \setdef{\phi' \in \CS{0}}{\phi'(1)=1}}
	\end{align}
	and any $\gamma \in \supp{\phi}$. We can split $\phi$ into two co-events using Lemma \ref{split_coevent_wrt_hist} as
	\begin{equation}
		\phi=\gamma^*\cdot \phi_1 + \phi_2 \,,
	\end{equation}
	where $\phi_1:=\dif{\phi}{\gamma}$, and we can consider the co-event
	\begin{equation}
		\phi' := \phi_1+\phi_2 \,.
	\end{equation}
	Now, suppose
	\begin{equation}
		\phi_1(P):=\dif{\phi}{\gamma}(P)=0 \quad \forall P \in \FDprec{0} \,. \label{dif_phi_on_prec_is_zero}
	\end{equation}
	Then for any $P \in \FDprec{0}$
	\begin{align}
		\phi'(P) &= \phi_1(P) + \phi_2(P)
		\\ &=\gamma^*(P) \cdot \phi_1(P) + \phi_2(P) \BY{dif_phi_on_prec_is_zero}
		\\ &= \phi(P)
		\\ &= 0 \qquad \text{since }\phi \in \Prec{0}\,.
	\end{align}
	Therefore, $\phi' \in \Prec{0}$. Moreover,
	\begin{align}
		\phi'(1) &= \phi_1(1)+\phi_2(1)
		\\ &= 1 \cdot\phi_1(1)+\phi_2(1)
		\\ &= \gamma^*(1) \cdot \phi_1(1)+\phi_2(1)
		\\ & = \phi(1)
		\\ &= 1 \,.
	\end{align}
	So $\phi'$ is also a preclusive co-event that affirms $1$. Moreover, we know from Lemma \ref{split_coevent_wrt_hist} that $\supp{\phi_1} \subset \supp{\phi}$ and $\supp{\phi_2} \subset \supp{\phi}$, meaning $\supp{\phi'} \subset \supp{\phi}$. Therefore, $\phi$ is not minimally supported amongst all preclusive co-events that affirm $1$, meaning it is not in $\BasicS{0}$, which is a contradiction. Therefore, \eqref{dif_phi_on_prec_is_zero} is false, meaning
	\begin{equation}
		\forall \phi \in \BasicS{0} \ \forall \gamma \in \supp{\phi} \ \exists P \in \FDprec{0} \st \dif{\CovS{0}}{\gamma}(P)=1 \,.
	\end{equation}
	And since, for $t>0$,
	\begin{equation}
		\BasicS{t} = \bigcup \limits_{\phi \in \BasicS{t-1}} \MinSupp[\big]{\PrecProl{t}(\phi)}
	\end{equation}
	we can apply Theorem \ref{dif_prec_one_iterate} iteratively to find that
	\begin{equation}
		 \forall \phi \in \BasicS{t} \ \forall \gamma \in \supp{\phi} \ \exists P \in \FDprec{t} \st \dif{\phi}{\gamma}(P)=1 \,. \label{basic_scheme_hist_prec}
	\end{equation}

	Now, consider any null event $P \in \FDprec{t}$. Then
	\begin{align}
		0 &= \mu_t(P)
		\\ &= {v_P}^\dagger d_t v_P \qquad \text{using \ref{decoherence_using_vectors}}
		\\ & = \sum \limits_{i,j} {a_j}^* {u_j}^\dagger d_t a_i u_i
		\\ & = \sum \limits_i \underbrace{\lambda_i}_{\geq 0} \underbrace{\vert a_i \vert^2}_{\geq 0} \,.
	\end{align}
	Where we expanded $v_P$ in terms of the Hermitian $d_t$'s orthogonal eigenvectors $u_i$ with expansion coefficients $a_i$, and $\lambda_i$ are the corresponding eigenvalues, which are non-negative by the definition of the decoherence matrix. Since all the terms in the sum are non-negative, the equality can only hold if $a_i=0$ when $\lambda_i \neq 0$ and vis versa. Therefore,
	\begin{align}
		d_t v_P &= \sum \limits_{i} a_i d_t u_i
		\\ &= \sum \limits_{i} a_i \lambda_i u_i
		\\ &= 0 \,, \label{d_on_P_is_zero}
	\end{align}
	where the last $0$ is the zero vector. Similarly, consider any null history $\gamma \in \HS{t}$, then
	\begin{equation}
		d_t v_\gamma = 0 \label{d_on_gamma_is_zero} \,.
	\end{equation}
	Now, consider the quantum measure for the event $(P+\gamma)$, given by
	\begin{align}
		\mu_t(P+\gamma) &= \begin{cases}
			(v_P - v_\gamma)^\dagger d_t (v_P - v_\gamma) & \quad \text{if }\gamma \in P
			\\ (v_P + v_\gamma)^\dagger d_t (v_P + v_\gamma) & \quad \text{if } \gamma \not \in P
		\end{cases}
		\\ &= 0 \BY{d_on_P_is_zero} \text{ and \eqref{d_on_gamma_is_zero}\,.}
	\end{align}
	Therefore, $(P+\gamma) \in \FDprec{t}$, and
	\begin{align}
		\dif{\phi}{\gamma}(P) &= \phi(P) + \phi(P+\gamma)
		\\ &= 0 + 0 \qquad \text{since } \phi \in \Prec{t} \,.
	\end{align}
	Note that this is true for all null $P$ and $\gamma$, so comparing this result with \eqref{basic_scheme_hist_prec} we can see that no null history can be in the support of any co-event in $\BasicS{t}$.
\end{proof}

Note that the proof would still work if we also included co-events $\CovS{0}$ outside of $\BasicS{0}$, but still with the property that
\begin{equation}
	\forall \gamma \in \CovS{0} \ \exists P \in \FDprec{0} \st \dif{\CovS{0}}{\gamma}(P)=1 \,.
\end{equation}

This proof has shown that, as we hoped, null \textit{histories}, which are non-interfering objects, do not hold any sway in our basic evolving scheme. However, null \textit{events} that contain more than one history can play a key role in determining the complexity of the co-events produced by the basic evolving scheme. The next theorem will show that, if at stage $t$ all the null events are only events that existed previously, up to some additional non-interfering null histories, then the basic evolving scheme's co-events will maintain the polynomial structure of the previous co-event. Therefore, increasing the structural complexity of an evolving co-event necessarily depends on the existence of new null events that arise through interference.
\begin{theorem} \label{no_new_precs}
	Consider $\CovS{t-1} \in \BasicS{t-1}$. Let $S^{(t-1)}=\supp[\big]{\CovS{t-1}}$ and write $\CovS{t-1}$ in terms of its polynomial expansion
	\begin{align}
		\CovS{t-1} &= \sum \limits_{i \in I} {E_i}^*
		\\ &= \sum \limits_{i \in I} \prod \limits_{\gamma \in E_i} \gamma^*\,,
	\end{align}
	where $I$ is some set of indices and each $E_i$ is unique. If the quantum measure at stage $t$ obeys
	\begin{equation}
		\mu_t(P)=0 \implies \exists P' \in \EA{t-1} \st  \mu_{t-1}(P')=0 \text{ and } \mu_t(\gamma)=0 \ \forall \gamma \in (P + \Ext{P'})\,,
	\end{equation}
	then
	\begin{equation}
		\MinSupp[\big]{\PrecProl{t}} = \setdef[\bigg]{ \sum \limits_{i \in I} \prod \limits_{\gamma \in E_i} {e(\gamma)}^*}{e: S^{(t-1)} \to \HS{t} \st {e(\gamma)}\Prev =\gamma \text{ and } \mu_t(e(\gamma)) \not = 0} \,.
	\end{equation}
\end{theorem}
\begin{proof}
	From Theorem \ref{prolongation_expansion_condition}, $\phi \in \Prol{t}$ iff its polynomial can be written in the form
	\begin{equation}
		\phi = \sum_{i\in I} {A_i}^* + \sum_{j\in J} ({B_j}^* + {C_j}^*) \,,
	\end{equation}
	where $J$ is some set of indices, each event is unique,
	\begin{listeqn}
		{A_i}\Prev = E_i
		\AND {B_i}\Prev = {C_i}\Prev \,.
	\end{listeqn}

	We now choose a map
	\begin{equation}
		\begin{aligned}
			e:S^{(t-1)} &\to \HS{t}
			\\ \gamma & \mapsto e(\gamma)
		\end{aligned}
	\end{equation}
	such that
	\begin{listeqn}
		{e(\gamma)}\Prev = \gamma
		\AND \mu_t\bigl(e(\gamma)\bigr) \not = 0 \,.
	\end{listeqn}
	Note that in order to find such a map we must be able to choose a non-null extension for every history in $S^{(t-1)}$. If all the extensions of $\gamma$ are null for some $\gamma \in S^{(t-1)}$, then
	\begin{align}
		\mu_{t-1}(\gamma) &= \mu_t\bigl( \Ext{\gamma} \bigr) \BY{qmeasure_consistency}
		\\ &= \sum \limits_{\gamma_1 \in \Ext{\gamma}} \sum \limits_{\gamma_2 \in \Ext{\gamma}} d_t(\gamma_1,\gamma_2)
		\\ &= 0 \qquad \text{since all extensions $\gamma_1,\gamma_2$ of $\gamma$ are null and by \eqref{null_means_no_interference}.}
	\end{align}
	But this contradicts with Theorem \ref{zero_hist_not_in_min_prec_prol_supp} because $\CovS{t-1} \in \BasicS{t-1}$. Therefore, there must be at least one non-null extension of $\gamma$ for all $\gamma \in S^{(t-1)}$, which means we can indeed find such a map $e$.

	If we then assign
	\begin{listeqn}
		A_i = \setdef{e(\gamma)}{\gamma \in E_i}
		\AND J = \emptyset\,,
	\end{listeqn}
	then
	\begin{equation}
		\supp{\phi} = \setdef[\big]{e(\gamma)}{\gamma \in S^{(t-1)}} \BYlem{supp_is_events_in_polynomial}.
	\end{equation}
	Moreover, consider any $P \in \FDprec{t}$. Then, $\exists P' \in \EA{t-1}$, where $P'$ is null, such that $P$ is made of $\Ext{P'}$ plus or minus some null histories. Therefore,
	\begin{align}
		\phi(P) &= \phi\bigl( \Ext{P'}+P+\Ext{P'} \bigr)\,.
		\intertext{But $(P+\Ext{P'})$ only contains null histories, which are not in our support by our choice of $e$ so we can use \eqref{independent_of_histories_outside_support} to remove each one of these null histories from $\phi$'s argument. Therefore,}
		\phi(P) &= \phi\bigl( \Ext{P'} \bigr)
		\\ &= \phi\Prev(P')
		\\ &= \CovS{t-1}(P') \qquad \text{since }\phi \in \Prol{t}
		\\ &= 0 \qquad \text{since } \CovS{t-1} \in \Prec{t}\,.
	\end{align}
	Therefore, $\phi \in \Prec{t}$. Also, note that the support can not be made any smaller whilst maintaining $\phi$ as a prolongation, otherwise we could not construct all the ${E_i}^*$ terms. Therefore, this $\phi$ is a minimally supported preclusive prolongation of $\CovS{t-1}$. Note that
	\begin{align}
		\sum \limits_{i \in I} \prod \limits_{\gamma \in E_i} {e(\gamma)}^* &= \sum \limits_{i \in I} {A_i}^*
		\\ &= \phi \,.
	\end{align}

	Now, consider the other $\phi \in \Prol{t}$. If they contain a null history in their support, then they are excluded from being minimally supported preclusive prolongations by Theorem \ref{zero_hist_not_in_min_prec_prol_supp}. Instead, we will now restrict ourselves to only consider supports that do not contain null histories.
	
	Once we have exhausted all the different choices for the $e$ maps, which essentially assign single non-null extensions to every history in the previous support, then the only other choices for $A_i$ is to choose two extensions $\gamma_1,\gamma_2 \in \HS{t}$ for a single $\gamma \in S^{(t-1)}$, i.e. $\gamma_1\Prev=\gamma_2\Prev=\gamma$, such that both $\gamma_1 \in \supp{\phi}$ and $\gamma_2 \in \supp{\phi}$. But then the support is larger than one we had previously, which means $\phi$ won't be a minimally supported preclusive prolongation.
	
	The other choice is to have a non-empty $J$ with
	\begin{align}
		{B_j}\Prev = {C_j}\Prev \,.
	\end{align}
	If we use histories to construct $B_j$ and $C_j$ that are outside the image of $e$ then we are again increasing our support more than necessary. If instead we try to only use histories in the image of $e$ then each history in the previous support is linked to a unique history in the new support, which means
	\begin{equation}
		 {B_j}\Prev = {C_j}\Prev \implies {B_j} = {C_j}\,.
	\end{equation}
	This excludes us from being able to build such pairs. There are therefore no other prolongations of $\CovS{t-1}$ that can be minimally supported.
\end{proof}

One immediate consequence of this is that the co-event chosen by the basic evolving scheme will not change if the new history space is just a copy of the previous one.
\begin{corollary}\label{same_hist_space}
	Consider $\CovS{t-1} \in \BasicS{t}$. If every $\gamma \in \HS{t-1}$ only has one extension in $\HS{t}$, i.e. $\HS{t-1}=\HS{t}$, then
	\begin{equation}
		\MinSupp[\big]{\PrecProl{t}} = \set[\big]{\CovS{t-1}}\,.
	\end{equation}
\end{corollary}
\begin{proof}
	Since $\HS{t}=\HS{t-1}$, by consistency all null events in $\EA{t}$ are extensions of null events in the equivalent $\EA{t-1}$, so we can apply Theorem \ref{no_new_precs}. Note we will use the polynomial expansion 
	\begin{equation}
		\CovS{t-1} = \sum \limits_{i \in I} {E_i}^*\,,
	\end{equation}
	and denote $S^{(t-1)}=\supp[\big]{\CovS{t-1}}$. The only map $\HS{t-1} \to \HS{t}$ with the property that histories are mapped to their extension is
	\begin{equation}
		\imath(\gamma)=\gamma \,.
	\end{equation}
	Therefore,
	\begin{align}
		\MinSupp[\big]{\PrecProl{t}} &= \setdef[\bigg]{ \sum \limits_{i \in I} \prod \limits_{\gamma \in E_i} {e(\gamma)}^*}{e: S^{(t-1)} \to \HS{t} \st {e(\gamma)}\Prev =\gamma \text{ and } \mu_t(e(\gamma)) \not = 0} 
		\\ &= \set[\bigg]{\sum \limits_{i \in I} \prod \limits_{\gamma \in E_i} {\imath(\gamma)}^*}
		\\ &= \set[\bigg]{\sum \limits_{i \in I} \prod \limits_{\gamma \in E_i} \gamma^*}
		\\ &= \set[\big]{\CovS{t-1}} \,.
	\end{align}
\end{proof}

Finally, we can also use Theorem \ref{no_new_precs} to confirm that the basic evolving scheme reduces to the classical scheme if it is supplied with a classical measure.
\begin{corollary}\label{return_to_classical}
	If $\mu_t$ is a classical measure for all $t$, that is
	\begin{equation}
		\mu_t(A+B)=\mu_t(A)+\mu_t(B) \label{classical_measure_rule} \,,
	\end{equation}
	then for every $t$
	\begin{equation}
		\BasicS{t} =\ClassicalS{t} \,.
	\end{equation}
\end{corollary}
\begin{proof}
	First note that if \eqref{classical_measure_rule} holds then
	\begin{equation}
		\mu_t(E)=0 \implies \mu_t(\gamma)=0 \quad \forall \gamma \in E \,.
	\end{equation}
	Therefore, all null events are made of null histories, which means Theorem \ref{no_new_precs} applies in this case.
	
	Now, consider the initial co-events in our basic evolving scheme. The above property of the quantum measure means, for $\gamma \in \HS{0}$, the co-event
	\begin{equation}
		\phi=\gamma^*
	\end{equation}
	will be preclusive iff $\gamma$ is not null. Moreover,
	\begin{equation}
		\gamma^*(1)=1 \,.
	\end{equation}
	Therefore, when $\gamma$ is not null, $\gamma^*$ is a preclusive co-event that affirms $1$. The only co-event that has a more minimal support is $0$, but $0(1)=0$, which means it can not be chosen for the initial co-event of our scheme. Therefore, for non-null $\gamma$
	\begin{equation}
		\gamma^* \in \BasicS{0}\,.
	\end{equation}
	Moreover, all other co-events are excluded from $\BasicS{0}$ because either they are not minimally supported, or because they contain null histories, which Theorem \ref{zero_hist_not_in_min_prec_prol_supp} excludes. Therefore,
	\begin{align}
		\BasicS{0} &= \setdef{\gamma^*}{\mu_0(\gamma) \neq 0}
		\\ &= \ClassicalS{0} \,.
	\end{align}

	Now, for $t>0$, suppose
	\begin{equation}
		\BasicS{t-1}=\ClassicalS{t-1}\,.
	\end{equation}
	Then
	\begin{align}
		\BasicS{t} &= \bigcup\limits_{\phi \in \ClassicalS{t-1}} \MinSupp[\big]{\PrecProl{t}(\phi) }
		\\ &= \bigcup\limits_{\substack{\gamma \in \HS{t-1} \st \\ \mu_{t-1}(\gamma) \neq 0}} \MinSupp[\big]{\PrecProl{t}(\gamma^*)}
		\\ &= \bigcup \limits_{\substack{\gamma \in \HS{t-1} \st \\ \mu_{t-1}(\gamma) \neq 0}} \setdef[\big]{ e(\gamma)^*}{e: \gamma \to \HS{t} \st {e(\gamma)}\Prev =\gamma \text{ and } \mu_t(e(\gamma)) \not = 0} \BYthm{no_new_precs}
		\\ &= \bigcup \limits_{\substack{\gamma \in \HS{t-1} \st \\ \mu_{t-1}(\gamma) \neq 0}} \setdef[\big]{ {\gamma'}^*}{\gamma' \in \HS{t}, \ \gamma'\Prev=\gamma \text{ and } \mu_t(\gamma') \neq 0}
		\\ &= \setdef[\big]{ { \gamma'}^*}{\gamma' \in \HS{t}, \ \mu_t(\gamma') \neq 0 \text{ and } \exists \gamma \in \HS{t-1} \st \gamma'\Prev=\gamma \text{ and } \mu_{t-1}(\gamma) \neq 0}
		\\ &= \setdef[\big]{ { \gamma'}^*}{\gamma' \in \HS{t}, \ \mu_t(\gamma') \neq 0 \text{ and } \mu_{t-1}(\gamma'\Prev) \neq 0}\,.
		\intertext{But for a classical measure, a non-null history will always have a non-null restriction. Therefore,}
		\BasicS{t} &= \setdef{{\gamma'}^*}{\gamma' \in \HS{t},\ \mu_t(\gamma') \neq 0}
		\\ &= \ClassicalS{t} \,.
	\end{align}
	The claim then follows by induction.
\end{proof}

\section{Other Evolving Schemes} \label{more_evolving_schemes}

Whilst the basic evolving scheme used above selects a small set of co-events from the total of $2^{2^{\HS{t}}}$ in $\CS{t}$, when applying the scheme to the $n$-site hopper it becomes apparent that $\BasicS{t}$ grows very quickly with $t$.

In practice, when computing co-events for the basic evolving scheme, once we have chosen $\CovS{t-1}$ we can use Lemma \ref{prec_prol_iff_supp_int} to generate a set of events that are the supports of co-events in $\PrecProl{t}$, and from here we can generate the set of events that are supports of co-events in $\MinSupp[\big]{\PrecProl{t}}$\footnote{The details of this operation are left out here, but the reader is free to contact the author for a method, including short cuts, if they are interested in computationally running this scheme.}. Once we have such an event $S$ we can generate all the co-events $ \phi \in \Basic{t}$ that have $S$ as their support using the restrictions
\begin{equation}
	\phi(E) = \begin{cases}
		\phi(E \cdot S) & \quad \text{if } E \not \subseteq S
		\\ 1 & \quad \text{if } \exists A \in \FA{t} \st A \cdot S = E
		\\ 0 & \quad \text{if } \exists D \in \FD{t} \st D \cdot S = E,
	\end{cases}
\end{equation}
and for the other events $E$ that do not match the above conditions we go through all possible choices of either mapping to zero or one. Each one of these co-events will be in $\Basic{t}$, so if there are $n$ events that do not match the above conditions there will be $2^n$ different co-events generated, all with the same support $S$.

Thus, one proposal from Rafael Sorkin is to only generate one co-event per support. In particular, we choose the co-event to be \textit{maximally affirmative}, that is we choose $\CovS{t}$ to be
\begin{equation}
	\CovS{t}(E) = \begin{cases}
		\CovS{t}(E \cdot S) & \quad \text{if } E \not \subseteq S
		\\ 1 & \quad \text{if } \exists A \in \FA{t} \st A \cdot S = E
		\\ 0 & \quad \text{if } \exists D \in \FD{t} \st D \cdot S = E
		\\ 1 & \quad \text{otherwise.}
	\end{cases}
\end{equation}
If we were to add this maximally affirmative condition on top of our basic evolving scheme, then the sequence of co-events generated $\CovS{t}$ would still be a sequence that the basic evolving scheme could also generate. Thus, if there are any properties that are true for \textit{all} co-events generated by the basic scheme then the maximally affirmative scheme would inherit these properties. In particular, Theorem \ref{zero_hist_not_in_min_prec_prol_supp}, Theorem \ref{no_new_precs} and Corollary \ref{same_hist_space} would still hold. In addition, in Corollary \ref{return_to_classical} all the co-events in $\MinSupp[\big]{\PrecProl{t}}$ are uniquely specified by their support, which means the maximally affirmative condition would not change the set of co-events once it is applied. Therefore, this corollary would still hold, which means the maximally affirmative scheme would also return the classical scheme.

An alternative way to alter the basic scheme in a way that reduces the number of co-events produced is to put additional pressure on the size of the co-events' supports. Already, Theorem \ref{no_new_precs} has given us the idea that the number of minimally supported preclusive prolongations, in this case the number of extension maps $e$, grows exponentially with the size of the previous support, so a small support here would reduce the number of choices. Since the size of the support can not shrink with prolongation, we require a constant downward pressure. One way to do this is to require the support to not only be minimally supported amongst all preclusive prolongations of $\CovS{t-1}$, but to require that it is also \textit{globally} minimally supported amongst all preclusive prolongations of \textit{any} previously allowed choices for $\CovS{t-1}$. This scheme would be given by generating
\begin{equation}
	\GlobalS{t} := \MinSupp[\bigg]{\ \bigcup \limits_{\phi \in \GlobalS{t-1}} \PrecProl{t}(\phi)}\,.
\end{equation}
Note that if $\GlobalS{t-1} \subseteq \BasicS{t-1}$, then
\begin{align}
	\GlobalS{t} &= \MinSupp[\bigg]{\ \bigcup \limits_{\phi \in \GlobalS{t-1}} \MinSupp[\big]{\PrecProl{t}(\phi)}} \label{global_scheme_apply_minsupp_twice}
	\\ & \subseteq \MinSupp[\bigg]{\ \bigcup \limits_{\phi \in \BasicS{t-1}} \MinSupp[\big]{\PrecProl{t}(\phi)}}
	\\ & \subseteq \bigcup \limits_{\phi \in \BasicS{t-1}} \MinSupp[\big]{\PrecProl{t}(\phi)}
	\\ &= \BasicS{t} \,.
\end{align}
Therefore, if we choose $\GlobalS{0} = \BasicS{0}$, then this globally minimal scheme would produce co-events that could also be produced by the basic evolving scheme. Therefore, again Theorem \ref{zero_hist_not_in_min_prec_prol_supp}, Theorem \ref{no_new_precs} and Corollary \ref{same_hist_space} would still hold. Moreover, from the identity in \eqref{global_scheme_apply_minsupp_twice} we can see that only one extra minimal support operation differentiates this scheme from the basic one. Therefore, if we note that the minimal support operation acting on a set of exclusively classical co-events does nothing to change that set, then we can see that, similar to Corollary \ref{return_to_classical}, when the globally minimal scheme is supplied with a classical measure it will also reduce to the classical scheme, and specifically \textit{not} a subset of it.

However, one potential problem for the globally minimal scheme is that a co-event in $\GlobalS{t-1}$ may not have any prolongations in $\GlobalS{t}$. If this were true then we would not be able to run a sequence of expressed co-events $\CovS{t}$ that are prolongations of the last without the danger of the process terminating, which would break the compatibility with the growing block view. This is clearly not a problem in the basic evolving scheme because the union over all previous co-events guarantees a prolongation for each of them, but with the stronger minimal support condition in the global scheme there seems to be a definite potential that some prolongations will be taken out entirely by the prolongations of a different co-event. However, an example of this problem is not yet known, and it may be possible to prove that such a termination is not actually possible (with, perhaps, some additional reasonable conditions).

\section{Discussion} \label{discussion}

We have introduced new objects and notation to describe a QMT for discrete evolving systems, along with Rafael Sorkin's co-event interpretation. We introduced the basic evolving scheme as a theory for extracting a set of allowed co-events by requiring them to be minimally supported preclusive prolongations of the previously chosen co-event. There are a number of properties that all the co-events produced by this scheme satisfy, and in particular we determined that null histories have no physical relevance in this theory since they are not contained in the support of any allowed co-event. In addition, we saw that if our scheme is given a classical measure it will give us the classical co-event scheme.

Now, there remain a number of other properties to explore within our basic evolving scheme, or the maximally affirmative and globally minimal schemes that follow from it. The first is whether the scheme can return classical co-events through a coarse graining procedure. In particular, if we partition our history space into histories that agree with different classical outcomes, and reduce our event algebra to events that are only made by adding together these classical partitions in different ways, then we would want the co-event acting on this reduced event algebra to act classically, i.e. classical logic would hold for this reduced algebra. One half-way point for establishing this as true would be to determine whether all the allowed co-events act classically on single time measurement outcomes, that is events that correspond to something we could measure in ordinary quantum mechanics, like ``the hopper is at site $i$ at time $\tau$''. However, for the basic and maximally affirmative scheme it is known that there exist allowed co-events for the $n$-site hopper that describe non-classical measurement outcomes. The known counter example had a relatively large support, so we might expect that a stricter minimality condition, like in the globally minimal scheme, would prevent the formation of such co-events. A number of counter examples were also found in a GHZ set up, but once the globally minimality condition was applied these counter examples fell away. 

In addition, we may want to consider the following questions for our schemes:
\begin{itemize}
	\item For a system like the $n$-site hopper, if we instead start with $\HS{0}'=\HS{T}$, i.e. we start our scheme at a later stage, then would $\BasicS{t}$ be approximately the same as ${\BasicS{t-T}}'$ at sufficiently large $t$? Clearly, if we choose ${\BasicS{0}}'=\BasicS{T}$ then the two would be exactly the same. But if this is quite different from the co-events that would result from the set in \eqref{basic_0_def}, and this difference persists into large $t$, then perhaps we would have to revise our choice of initial co-events. Note that this is not a problem for the classical scheme because $\ClassicalS{t}$ is always made of all the classical co-events who's supporting history is not null, and so starting the scheme at a later stage has no effect.
	\item If our system is composed of isolated sub systems (with or without entanglement between them), then does evolving them all simultaneously differ from evolving them separately in sequence? In a GHZ set up, we found that the globally minimal and maximally affirmative schemes produced different results if we evolved lab 1 by one step, then lab 2 by one step, then lab 3 by one step as opposed to evolving all labs by one step simultaneously. This leaves the question of which choice is the correct one, and really we want a scheme that does not suffer from these ambiguities.
	\item If we have two disjoint systems $\mathbf{A}$ and $\mathbf{B}$, and we apply the scheme to their union, can this scheme reduce to the scheme applied to just $\mathbf{A}$? For two disjoint systems the history space would be the product of the two individual history spaces and the decoherence matrix would be given by the tensor product of the individual decoherence matrices. If the local reality of $\mathbf{A}$ looked different when $\mathbf{B}$ is included then our scheme would be inherently non-local, even without any entanglement.
	\item If there exists an event that corresponds to something we could measure in ordinary quantum mechanics with non-zero probability, does there exist a co-event produced by our scheme that affirms it? In other words, if a measurement outcome event is not impossible we expect there to exist a possible co-event where that event happens.
\end{itemize}
These questions are challenging and will likely further define and alter what evolving scheme is appropriate for producing the realities we expect. In order to keep the eventual evolving scheme light and insightful, we hope to be able to answer the above questions positively using basic primitive rules, as we did to prove Theorem \ref{zero_hist_not_in_min_prec_prol_supp}, as opposed to adding in complex rules that trivially give the desired results.

\subsection*{Acknowledgements}
We thank Rafael Sorkin for the idea of the evolving co-event scheme and discussions of the excluded null histories claim. We thank Fay Dowker for discussions related to all the claims in this paper. Henry Wilkes is supported by STFC grant ST/N504336/1. This research was supported in part by Perimeter Institute for Theoretical Physics. Research at Perimeter Institute is supported by the Government of Canada through the Department of Innovation, Science and Economic Development and by the Province of Ontario through the Ministry of Research and Innovation.

\end{document}